\documentclass[10pt,aps,pra,twocolumn,floatfix,nofootinbib,superscriptaddress,longbibliography]{revtex4-2}
\usepackage{placeins}

\usepackage{bm,graphicx,mathrsfs,amsmath,amssymb,mathtools,makecell,bbm, amsthm,dsfont,color,nicefrac,framed,enumitem,tikz,physics,wrapfig,amsfonts,tcolorbox,times,txfonts,lipsum,nicematrix}
\usepackage{float}
\usepackage{comment}

\renewcommand{\ket}[1]{\lvert #1\rangle}

\newcommand{\iden}{\mathbbm{1}}
\newcommand{\id}{\mathbbm{1}}
\newcommand{\Wtwo}{\mathcal{W}_2}
\newcommand{\Qbr}{\mathcal{Q}_\mathrm{br}}

\definecolor{alecolor}{RGB}{110,190,110}      % verde um pouco mais escuro
\definecolor{changescolor}{RGB}{0,120,0}

\definecolor{mycitecolor}{RGB}{204,170,0}     % citações em amarelo
\definecolor{mylinkcolor}{RGB}{110,190,110}   % verde mais escuro
\definecolor{myurlcolor}{RGB}{110,190,110}    % URLs em verde mais escuro

\usepackage[
    colorlinks=true,
    citecolor=mylinkcolor,
    linkcolor=mycitecolor,
    urlcolor=myurlcolor
]{hyperref}
\usepackage{ragged2e}
\definecolor{equationcolor}{RGB}{154,135,198}
\usetikzlibrary{arrows.meta}
\usetikzlibrary{positioning}
\usetikzlibrary{calc}

\newtheorem{theorem}{Theorem}
\newtheorem{definition}[theorem]{Definition}
\newtheorem{corollary}[theorem]{Corollary}
\newtheorem{lemma}[theorem]{Lemma}
\newtheorem{proposition}[theorem]{Proposition}
\theoremstyle{remark}
\newtheorem{remark}[theorem]{Remark}
\newcommand{\WQ}{W_{Q}}
\newcommand{\Hmin}{H_{\min}}
\newcommand{\B}{\mathcal{B}}
\newcommand{\corr}[3]{\langle B_{#1}C_{#2}\rangle^{[#3]}}
\newcommand{\avg}[2]{\langle #1\rangle^{[#2]}}
\providecommand{\tr}{\operatorname{tr}}

\begin{document}
\title{Randomness Certification and Trade-offs in the Prepare-and-Broadcast Scenario}

\author{Tailan S. Sarubi}
\affiliation{Physics Department, Federal University of Rio Grande do Norte, Natal, 59072-970, Rio Grande do Norte, Brazil}
\affiliation{International Institute of Physics, Federal University of Rio Grande do Norte, 59078-970, Natal, Brazil}
\author{Moisés Alves}
\affiliation{Physics Department, Federal University of Rio Grande do Norte, Natal, 59072-970, Rio Grande do Norte, Brazil}
\affiliation{International Institute of Physics, Federal University of Rio Grande do Norte, 59078-970, Natal, Brazil}
\affiliation{QuIIN - Quantum Industrial Innovation, EMBRAPII CIMATEC Competence Center in Quantum Technologies, SENAI CIMATEC, Av. Orlando Gomes 1845, 41650-010, Salvador, BA, Brazil. }

\author{Vinícius F. Alves}
\affiliation{Physics Department, Federal University of Rio Grande do Norte, Natal, 59072-970, Rio Grande do Norte, Brazil}
\affiliation{International Institute of Physics, Federal University of Rio Grande do Norte, 59078-970, Natal, Brazil}
\author{Santiago Zamora}
\affiliation{International Institute of Physics, Federal University of Rio Grande do Norte, 59078-970, Natal, Brazil}
\author{Diogo H. G. Duarte}
\affiliation{Physics Department, Federal University of Rio Grande do Norte, Natal, 59072-970, Rio Grande do Norte, Brazil}
\affiliation{International Institute of Physics, Federal University of Rio Grande do Norte, 59078-970, Natal, Brazil}
\author{A.~de~Oliveira~Junior}
\affiliation{Center for Macroscopic Quantum States bigQ, Department of Physics,
Technical University of Denmark, Fysikvej 307, 2800 Kgs. Lyngby, Denmark}
\author{Rafael Chaves}
\affiliation{International Institute of Physics, Federal University of Rio Grande do Norte, 59078-970, Natal, Brazil}

\begin{abstract}
We investigate the prepare-and-broadcast scenario, a multipartite extension of dimension-constrained prepare-and-measure experiments in which a quantum system is distributed to multiple receivers. We derive fundamental trade-offs between prepare-and-measure witnesses, Bell nonlocality, and quantum random access code performance. We further develop a semi-device-independent randomness certification framework based on prepare-and-broadcast witnesses, showing that the maximal quantum violation certifies two bits of joint randomness, exceeding the limit achievable from the CHSH inequality while remaining robust to noise. Finally, we show that the prepare-and-broadcast scenario naturally accommodates stronger adversarial models in which the eavesdropper retains a quantum system correlated with the measurement device, providing a natural framework for semi-device-independent randomness certification against quantum side information.
\end{abstract}

\maketitle

\section{Introduction}

Quantum communication exploits the properties of quantum systems to accomplish information-processing tasks beyond the capabilities of classical communication. A defining feature of quantum information is that unknown quantum states cannot be perfectly copied \cite{scarani2005quantum}. While this limitation underlies important applications such as secure quantum communication, it also fundamentally constrains how the information encoded in a single quantum system can be distributed among multiple users. Understanding the capabilities and limitations imposed by this restriction is therefore a central question in quantum communication \cite{fan2014quantum}.

Dimension-constrained prepare-and-measure (PAM) scenarios have proved to be a particularly successful framework for addressing such questions \cite{bohr2026quantum}. In these scenarios, a preparation device encodes information into a physical system of bounded dimension that is transmitted to a measurement device. Comparing the correlations achievable with classical and quantum systems under the same communication constraint has led to fundamental results on dimension witnesses ~\cite{gallego2010device,tavakoli2015quantum,chaves2018causal,Poderini2020}, quantum random access codes (QRACs) ~\cite{Ambainis2002,Nayak1999,Ambainis2008QRAC,pawlowski2010entanglement}, entanglement certification~\cite{tavakoli2018semi,moreno2021semi,pauwels2022entanglement,vieira2023interplays}, non-stabilizerness certification~\cite{zamora2025semi}, and semi-device-independent information processing \cite{pawlowski2009information,pawlowski2011semi,li2012semi,woodhead2014imperfections,lunghi2015self,passaro2015optimal,chaves2015information,tavakoli2021correlations,alves2026semi}. Owing to its minimal physical assumptions, the PAM scenario has also become a standard framework for certifying nonclassicality and randomness. However, because in standard PAM scenarios the communicated system is delivered to a single receiver, its framework does not address how the information carried by a quantum system can be distributed among multiple observers.

To address this limitation, the prepare-and-broadcast (PAB) scenario was introduced \cite{wang2019characterising,ioannou2022receiver,jia2025characterizing,sarubi2026prepare} as a natural multipartite extension of the PAM framework. Instead of transmitting the communicated system directly to a single receiver, it is processed by a broadcast channel whose outputs are distributed to multiple observers. This combines features of prepare-and-measure and Bell scenarios within a single causal structure: each receiver defines a marginal prepare-and-measure experiment, while their joint conditional statistics exhibit Bell correlations. Early formulations of the PAB scenario relied on partial knowledge of the prepared states \cite{wang2019characterising,ioannou2022receiver,jia2025characterizing}, whereas more recently a fully semi-device-independent formulation was introduced in which the communication dimension is the only assumption \cite{sarubi2026prepare}. Within this framework, the classical, quantum, and nonsignaling correlation sets were characterized, and certification methods and activation phenomena were proposed.

Having established the framework, a natural next step is to understand how quantum communication advantages can be distributed. Since the information encoded in a single quantum system must now be shared among multiple receivers, one expects fundamental trade-offs imposed by quantum mechanics. Moreover, the coexistence of prepare-and-measure and Bell correlations within the same experiment suggests new connections between different manifestations of nonclassicality. In this work, we investigate these trade-offs and develop randomness certification in the prepare-and-broadcast scenario.

Our first objective is to characterize how quantum communication advantages can be distributed in the prepare-and-broadcast scenario. To this end, we investigate trade-offs between violations of prepare-and-measure witnesses in the marginal behaviors, Bell inequality violations in the conditional correlations, and the prepare-and-measure advantages attainable by different receivers. While different manifestations of nonclassicality can coexist, we show that they obey quantitative trade-offs. In particular, since the four-preparation witness is directly related to the success probability of the $2\rightarrow1$ quantum random access code \cite{Ambainis2008QRAC}, these trade-offs admit a natural operational interpretation as limitations on distributing a quantum communication advantage. Interestingly, the optimal boundary coincides with that of asymmetric phase-covariant quantum cloning~\cite{Bruss2000,DAriano2003}, establishing a direct connection between prepare-and-broadcast communication and optimal quantum cloning.

We also investigate randomness certification in the prepare-and-broadcast scenario. Assuming bounded communication dimension and considering different models of an eavesdropper, we develop a semi-device-independent randomness certification protocol based on prepare-and-broadcast witnesses. Using semidefinite-programming relaxations, we bound the adversary's guessing probability and quantify the amount of certifiable randomness as a function of the observed witness value. Remarkably, for the prepare-and-broadcast witness considered here, the maximal quantum violation certifies two bits of joint output randomness, exceeding the maximum of approximately $1.23$ bits obtainable from the paradigmatic CHSH inequality \cite{Pironio2010}. Moreover, this maximal randomness can be certified robustly, tolerating significantly higher levels of noise than alternative protocols capable of certifying two bits of randomness, such as those based on tilted Bell inequalities \cite{acin2012randomness}.

The paper is organized as follows. In Sec. \ref{sec:thepampab}, we introduce the PAM and PAB scenarios. In particular, we derive a necessary and sufficient criterion for nonclassicality in the four-preparation PAM scenario associated with the $2\rightarrow1$ random access code, generalizing the results of Ref. \cite{Poderini2020}. In Sec. \ref{sec:broadcast_qrac_tradeoff}, we investigate trade-offs between violations of PAM witnesses and Bell inequalities, and establish their connection with the optimal trade-off for the $2\rightarrow1$ QRAC obtained from asymmetric phase-covariant quantum cloning. In Sec. \ref{sec:randomness}, we study semi-device-independent randomness certification in both the PAM and PAB scenarios under the assumption of classical side information. We then employ the PAB scenario to strengthen the adversarial model in Sec. \ref{sec:qeve} by allowing the eavesdropper to retain a quantum system correlated with the measurement device. Finally, in Sec. \ref{sec:discussion}, we discuss our findings, while Appendices \ref{app:A}, \ref{app:support_function_lower_bounds} and \ref{app:endpoint} contain technical details and proofs.

\section{The prepare-and-measure and prepare-and-broadcast scenarios}\label{sec:thepampab}

The prepare-and-measure scenario, illustrated in Fig.~\ref{fig:PAM_scenario}, consists of a preparation device and a measurement device \cite{gallego2010device}. In each run, the preparation device receives an input $x\in\mathcal X$ and encodes it into a physical system, which is sent to the measurement device. The latter receives an input $y\in\mathcal Y$ and produces an output $b\in\mathcal B$, giving rise to the conditional probability distribution $p(b|x,y)$.

In a classical realization, the communicated system is represented by a message $m\in\mathcal M$, while in the quantum case it is represented by a quantum state $\rho_x$. We also allow the devices to share classical randomness described by a hidden variable $\lambda$, assumed to be independent of the freely chosen inputs. The corresponding classical behaviors admit the decomposition
\begin{equation}
p(b|x,y)=\sum_{\lambda,m} p(\lambda)\, p(m|x,\lambda)\, p(b|m,y,\lambda).
\label{eq:classical_pam}
\end{equation}

Without further restrictions, this model is trivial, since the preparation device could encode the entire input $x$ into the communicated message. Throughout this work we therefore impose a dimension constraint, requiring $|\mathcal M| = d < |\mathcal X|$. For fixed input and output alphabets, the resulting classical behaviors form a convex polytope whose facets define classical dimension witnesses.

In the quantum model, the communicated message is replaced by a $d$-dimensional quantum system. For each input $x$, the preparation device emits a state $\rho_x\in\mathcal D(\mathbb C^d)$, while for each measurement choice $y$ the measurement device performs a POVM $\{M_{b|y}\}_b$. The observed statistics are then given by Born's rule, $p(b|x,y)=\mathrm{tr}\!\left(\rho_xM_{b|y}\right)$. We focus on the binary PAM scenarios with $|\mathcal M|=|\mathcal Y|=|\mathcal B|=2$, and three or four preparations, that is, $|\mathcal X|=3,4$. Given a set of qubit preparations, the relevant question is whether there exist binary measurements for which the resulting behavior lies outside the classical binary-message polytope. Equivalently, we ask whether the preparations admit a classical binary-message model.

\begin{figure}[t]
\centering
\includegraphics{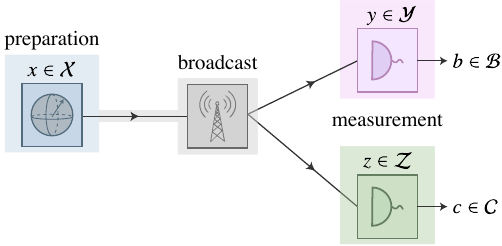}
\caption{\textbf{PAB scenario.} In each run, the preparation device receives an input $x\in\mathcal X$, prepares a quantum state that undergoes a broadcasting map that distributes it to two spatially separated measurement devices, receiving inputs $y\in\mathcal Y$ and $z\in\mathcal Z$ and producing local outputs $b\in\mathcal B$ and $c\in\mathcal C$, respectively. The PAB scenario is thus described by a probability distribution $p(b,c\vert x,y,z)$. Notice that if one of the parties is traced out, for instance, Charlie, one recovers the usual PAM scenario, described by a probability distribution $p(b\vert x,y)$}
\label{fig:PAM_scenario}
\end{figure}

For three preparations, the classical polytope is characterized (up to relabelings) by the facet
\begin{equation}
S_{3}=E_{00}+E_{01}+E_{10}-E_{11}-E_{20}\le3,
\label{eq:S3}
\end{equation}
where $E_{xy}=p(0|x,y)-p(1|x,y)$.

Optimizing over Bob's measurements, Ref.~\cite{Poderini2020} showed that a triple of qubit preparations is classically compatible if and only if
\begin{equation}
\max_{(a,b,c)\in\{(0,1,2),(0,2,1),(2,1,0)\}} \left( \|\vec r_a+\vec r_b-\vec r_c\|_2 + \|\vec r_a-\vec r_b\|_2 \right) \le3,
\label{eq:criterion3prep}
\end{equation}
where $r_i$ stands for the Bloch vector representing the prepared quantum state.

For four preparations, the classical polytope contains, besides embedded copies of Eq.~\eqref{eq:S3}, the genuinely four-preparation facet
\begin{equation}
S_{4} = E_{00}+E_{01}+E_{10}-E_{11} -E_{20}+E_{21} -E_{30}-E_{31} \le4.
\label{eq:S4}
\end{equation}
Our first result is to extend the characterization of Ref.~\cite{Poderini2020} to this scenario. Defining
\begin{equation}
\Delta_{ij|kl} = \| \vec r_i+\vec r_j-\vec r_k-\vec r_l \|_2,
\end{equation}
and optimizing Eq.~\eqref{eq:S4} over Bob's measurements yields the independent conditions
\begin{align}
\Delta_{01|23}+\Delta_{02|13}&\le4,\nonumber\\
\Delta_{01|23}+\Delta_{12|03}&\le4,\label{eq:criterion4prepA}\\
\Delta_{02|13}+\Delta_{12|03}&\le4,
\end{align}
which correspond to the inequivalent relabelings of the four-preparation facet (see Appendix~\ref{app:A}).

In addition, every subset of three preparations must satisfy the three-preparation criterion~\eqref{eq:criterion3prep},
\begin{equation}
\max_{\substack{ a,b,c\in\{0,1,2,3\}\\ a,b,c\ \mathrm{distinct} }} \left( \|\vec r_a+\vec r_b-\vec r_c\|_2 + \|\vec r_a-\vec r_b\|_2 \right) \le3.
\label{eq:criterion4prepB}
\end{equation}
Together, Eqs.~\eqref{eq:criterion4prepA} and \eqref{eq:criterion4prepB} provide a necessary and sufficient geometric criterion for classical compatibility in the four-preparation binary PAM scenario.

The prepare-and-broadcast scenario \cite{sarubi2026prepare}, also illustrated in Fig.~\ref{fig:PAM_scenario}, extends the standard prepare-and-measure setting by replacing the single receiver with several receivers. Throughout this work, we consider the simplest nontrivial case of two receivers, Bob and Charlie. Alice receives an input $x$, prepares a $d$-dimensional physical system, and sends it to a broadcasting device. The broadcaster processes the incoming system and distributes two outputs to Bob and Charlie, who independently choose measurement settings $y$ and $z$ and produce outcomes $b$ and $c$. The observed behavior is therefore described by the conditional probabilities $p(b,c|x,y,z)$.

The PAB scenario naturally combines features of prepare-and-measure and Bell experiments. By marginalizing over one receiver, for example, $p(b|x,y)=\sum_c p(b,c|x,y,z)$, one recovers an ordinary PAM experiment. Conversely, if the preparation input is fixed, the remaining correlations between Bob and Charlie correspond to a standard Bell scenario. Thus, a single PAB experiment simultaneously gives access to communication witnesses through the marginal behaviors and Bell nonlocality through the conditional joint correlations.

As in Ref.~\cite{sarubi2026prepare}, by using two labels, we distinguish different classes of resources according to the nature of Alice's communication and the broadcasting device. The first label specifies whether Alice sends a classical ($\mathcal C$) or quantum ($\mathcal Q$) system, while the second specifies whether the broadcaster is classical, quantum, or nonsignalling. The fully classical model ($\mathcal{CC}$) is therefore described by
\begin{equation}
p_{\mathcal{CC}}(b,c|x,y,z) = \sum_{m,\lambda} p(\lambda) D_\lambda^A(m|x) D_\lambda^B(b|y,m) D_\lambda^C(c|z,m),
\label{eq:cc_model}
\end{equation}
where $m$ is a classical message of bounded cardinality $|m|=d$, $\lambda$ denotes shared randomness, and the response functions are deterministic without loss of generality.

Replacing the classical message by a quantum state and allowing for a quantum broadcaster yields the fully quantum model ($\mathcal{QQ}$),
\begin{equation}
p_{\mathcal{QQ}}(b,c|x,y,z) = \sum_\lambda p(\lambda) \, \mathrm{Tr} \!\left[ \mathcal T_\lambda(\rho_{x,\lambda}) \left( M^B_{b|y}\otimes M^C_{c|z} \right) \right],
\label{eq:qq_model}
\end{equation}
where $\rho_{x,\lambda}$ is a $d$-dimensional quantum state and $\mathcal T_\lambda$ is a quantum broadcast channel. Hybrid models, such as $\mathcal{CQ}$ and $\mathcal{CNS}$, are obtained by allowing the broadcaster to distribute quantum or nonsignalling correlations while Alice is restricted to classical communication. Their precise definitions and properties can be found in Ref.~\cite{sarubi2026prepare}.

To our goals in Sec. \ref{sec:randomness} regarding the certification of randomness, we will consider the PAM inequalities \eqref{eq:S3} and \eqref{eq:S4} as well as a PAB inequality derived in  Ref.~\cite{sarubi2025pab}, given by
\begin{align}
W_{\mathcal{PAB}} = {}& 2\langle B_0C_0\rangle^{[0]}+2\langle B_1C_1\rangle^{[0]} +2\langle B_0C_1\rangle^{[1]} \nonumber\\ & -2\langle B_1C_0\rangle^{[1]}-\langle B_0C_0\rangle^{[2]} -\langle B_0C_1\rangle^{[2]} \nonumber\\ &+\langle B_1C_0\rangle^{[2]} -\langle B_1C_1\rangle^{[2]} \le_{CC} 6 \le_{QQ} 8+2\sqrt{2},
\label{eq:CC2}
\end{align}
where the subscripts $CC$ and $QQ$ refer to the maximum values achieved in purely classical or fully quantum realization of the PAB scenario. 

\section{Prepare and Broadcast trade-offs}
\label{sec:broadcast_qrac_tradeoff}

\begin{figure*}[t!]
\includegraphics{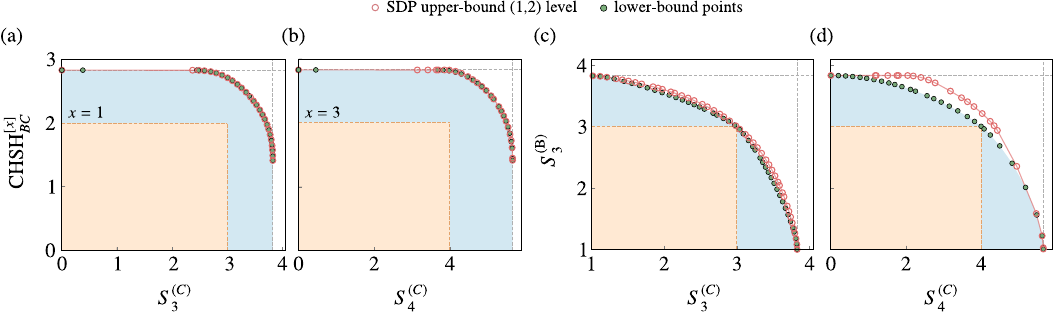} 
\caption{\textbf{Trade-offs between PAM and Bell witnesses in the PAB scenario.} Panels (a) and (b) show the trade-off between Charlie's marginal PAM witness and the CHSH value evaluated by Bob and Charlie conditional on Alice's preparation input: (a) $\mathrm{CHSH}_{BC}^{[1]}$ versus $S_3^{(C)}$, and (b) $\mathrm{CHSH}_{BC}^{[3]}$ versus $S_4^{(C)}$. For $S_3$, the choice $x=1$ is representative of the equivalence class $x\in\{0,1\}$. For $S_4$, all four conditioning inputs are equivalent under relabelings, and $x=3$ is chosen as a representative. Panels (c) and (d) show marginal receiver--receiver trade-offs: (c) $S_3^{(B)}$ versus $S_3^{(C)}$, and (d) $S_3^{(B)}$ versus $S_4^{(C)}$.  The pale-orange rectangles mark the regions in which both plotted quantities obey their classical bounds, namely $\mathrm{CHSH}\leq 2$, $S_3\leq 3$, and $S_4\leq 4$, as applicable. The pale-blue areas are the portions of the regions certified by the explicit qubit strategies that lie outside the corresponding classical rectangles. The unshaded regions between the green lower bounds and red SDP upper bounds in panels (c) and (d) represent the remaining numerical uncertainty in the optimal trade-off. In panels (a) and (b), the achievable region includes points at which both the marginal PAM inequality and the conditioned CHSH inequality are violated. In panel (d), the gap between the lower and upper bounds near $(S_4^{(C)},S_3^{(B)})=(4,3)$ means that simultaneous violations are not excluded by the present relaxation.  Grey dashed lines mark the standalone quantum maxima of the individual witnesses: $2\sqrt{2}$ for CHSH, $1+2\sqrt{2}$ for $S_3$, and $4\sqrt{2}$ for $S_4$. Green filled circles are lower bounds obtained from explicit qubit strategies, while red open circles are SDP upper bounds from the $(1,2)$ level of the SDP relaxation proposed in Ref.~\cite{sarubi2026prepare}.}
\label{fig:combined_tradeoffs}
\end{figure*}

As discussed above, the PAB scenario contains both PAM and Bell scenarios as special cases. This makes it natural to ask how the corresponding forms of nonclassicality constrain each other. In this section, we study trade-offs between violations of marginal PAM witnesses and violations of Bell inequalities between Bob and Charlie. Operationally, this amounts to asking how the information carried by a single qubit can be distributed between two receivers through a broadcasting channel. 

Explicit qubit strategies give lower bounds, obtained with the support-function method described in Appendix.~\ref{app:support_function_lower_bounds},  while the SDP relaxation introduced in Ref. \cite{sarubi2026prepare} gives upper bounds. When the two coincide within numerical precision, we regard the corresponding boundary as certified within the relaxation used here. We consider two types of trade-offs. First, we compare marginal PAM witnesses with conditioned Bell violations between Bob and Charlie. We find that PAM and Bell violations can coexist, but only up to a point; namely, increasing the value of the PAM witness reduces the maximal Bell violation that can be achieved. Second, we compare the PAM witnesses obtained from Bob's and Charlie's marginal behaviors. In this case, the trade-off is stronger. A violation of the relevant PAM witness by one receiver prevents a violation of the corresponding marginal PAM witness by the other.

We begin with the trade-off between PAM and Bell inequalities. For a fixed preparation input $x$, the behavior $p(b,c\vert x,y,z)$ defines a standard bipartite Bell scenario between Bob and Charlie. We denote by $\mathrm{CHSH}^{[x]}$ the CHSH expression evaluated on this conditional behavior. Since Bob and Charlie each have two dichotomic measurements, up to relabelings, the CHSH inequality is the only nontrivial Bell inequality in this conditional scenario~\cite{clauser1969proposed}.

For the $S_3$ witness, the inputs do not all play the same role. This gives two inequivalent classes of curves: one for $x=0,1$ and another for $x=2$. For simplicity, in Fig.~\ref{fig:combined_tradeoffs} we focus on the case $x=0$. For $S_4$, the four inputs are treated equivalently up to relabelings of the facet, and the corresponding conditioned CHSH trade-offs coincide. The curves show that it is possible to violate a marginal PAM witness and a conditioned CHSH inequality at the same time. From an operational point of view, this means that the same broadcast realization can provide a quantum advantage in a PAM-type task and nonclassical correlations between Bob and Charlie. For example, the PAM task can be interpreted as a random access code \cite{Ambainis2002,Nayak1999,Ambainis2008QRAC}, while Bell violations are linked to advantages in distributed information-processing tasks \cite{brukner2004bell}. However, the coexistence is limited. In all cases, we find a threshold value of the PAM witness beyond which the conditioned CHSH inequality can no longer be violated.

We then consider the trade-off between the marginal PAM witnesses of Bob and Charlie. Here the behavior is different. As shown in Fig.~\ref{fig:combined_tradeoffs}, if Bob violates the relevant $S_3$ or $S_4$ witness, then Charlie's corresponding marginal witness remains below its classical bound, and conversely. Thus, while the broadcast system can distribute some information to both receivers, the part of the information responsible for a PAM advantage cannot be copied into two simultaneous marginal PAM violations.

In the next subsection, we analyze this effect in more detail for the $S_4$ witness. This case has a direct interpretation in terms of the success probability of the $2\to1$ QRAC and connects naturally with asymmetric phase-covariant cloning.

\subsection{Broadcasting, random access codes and phase-covariant cloners}

Random access codes (RAC) provide a standard operational way of quantifying how well classical information can be encoded into a dimension-limited physical system and later partially decoded. In the usual $2\to1$ quantum random access code, Alice encodes two bits $x=(x_0,x_1)$ into a single quantum system and sends it to a receiver. The receiver is asked to guess one of the two bits. The figure of merit is the average success probability, namely the probability that the receiver correctly outputs the requested bit, averaged over all possible inputs $x\in\{(0,0),(0,1),(1,0),(1,1)\}$ and over the two possible decoding tasks $y\in\{0,1\}$. A classical bit reaches the optimal success probability $3/4$, while a qubit reaches $1/2+1/(2\sqrt2)$~\cite{Ambainis2002,Nayak1999,Ambainis2008QRAC}.

In the PAB scenario, the natural question is how this qubit advantage can be distributed between two receivers. We label Alice's four preparations by $x\in\{0,1,2,3\}$ and identify each label with a two-bit string $x\leftrightarrow(x_0,x_1)$: $0\leftrightarrow(0,0)$, $1\leftrightarrow(0,1)$, $2\leftrightarrow(1,0)$, $3\leftrightarrow(1,1)$. Bob receives a query $y\in\{0,1\}$ and tries to output the bit $b=x_y$, while Charlie receives a query $z\in\{0,1\}$ and tries to output $c=x_z$. In the quantum broadcasting model of Eq.~\eqref{eq:qq_model}, Alice's qubit message is mapped to a bipartite state $\rho_x^{BC}=\mathcal T(\rho_x)$ which Bob and Charlie measure locally. For Bob, the $2\to1$ quantum random access code score is
\begin{align}
P_B &= \frac18 \sum_{x=0}^3\sum_{y=0}^1  p(b=x_y\vert x,y) \nonumber\\ &= \frac12 + \frac{1}{16} \sum_{x,y} (-1)^{x_y} \langle B_y\rangle^{[x]},
\label{eq:qrac_PB_marginal}
\end{align}
where $p_B(b\vert x,y)$ is Bob's marginal behavior. Charlie's score $P_C$ is defined analogously, replacing $B_y$, $b$, and $y$ by $C_z$, $c$, and $z$ respectively. Comparing Eq.~\eqref{eq:qrac_PB_marginal} with the definition of the $S_4$ witness gives
\begin{align}
    P_B&=\frac12+\frac{S_4^{(B)}}{16},\\
    P_C&=\frac12+\frac{S_4^{(C)}}{16}.
    \label{eq:s4_qrac_probability_relation}
\end{align}
Thus, the marginal $S_4$ witnesses have a direct QRAC interpretation. Values above $S_4=4$ correspond to success probabilities above the optimal classical RAC value $\tfrac34$. Classically, a single bit can be copied and distributed to both receivers. Therefore Bob and Charlie can simultaneously reach $P_B=P_C=\tfrac34$, or equivalently $S_4^{(B)}=S_4^{(C)}=4$. The question is whether the quantum advantage above this copyable classical point can also be shared by the two marginals.

The boundary shown in Fig.~\ref{fig:broadcast_qrac_tradeoff} answers this question. It is obtained with the support-function method of Appendix~\ref{app:support_function_lower_bounds}. The plotted curve combines explicit lower-bound strategies with the corresponding upper bounds from the SDP relaxation of Ref.~\cite{sarubi2026prepare}. The two agree within numerical precision and give
\begin{equation}
\left(P_B-\frac12\right)^2 + \left(P_C-\frac12\right)^2 = \frac18.
    \label{eq:phase_covariant_qrac_curve}
\end{equation}
Equivalently, $\big(S_4^{(B)}\big)^2+\big(S_4^{(C)}\big)^2=32$. The symmetric point of Eq.~\eqref{eq:phase_covariant_qrac_curve} is exactly the classical copyable point, $P_B=P_C=\frac34$. Thus, the quantum advantage of the $2\to1$ QRAC is not shareable in this broadcast setting. If Bob goes above the classical value, Charlie must go below it, and vice versa. At the endpoint $P_B=\frac12+\frac{1}{2\sqrt2}$,  Bob reaches the optimal qubit QRAC value, while Eq.~\eqref{eq:phase_covariant_qrac_curve} forces $P_C=\tfrac12$. In terms of the PAM witnesses, this corresponds to $S_4^{(B)}=4\sqrt2$ and $S_4^{(C)}=0$. 

\begin{figure}
    \centering
    \includegraphics{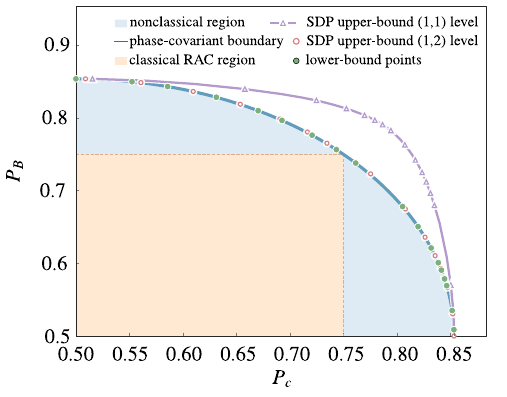}
    \caption{\textbf{Marginal QRAC trade-off in the PAB scenario.} The curve shows the possible pairs of success probabilities $(P_B,P_C)$ for Bob and Charlie in the broadcast $2\to1$ QRAC. The inner approximation obtained from explicit strategies through the support-function method coincides, up to numerical precision, with the outer approximation given by the $(1,2)$ level of the SDP relaxation. The boundary is the same as the trade-off for optimal asymmetric phase-covariant cloning: if one receiver beats the classical QRAC value $\tfrac34$, the other must fall below it.}
    \label{fig:broadcast_qrac_tradeoff}
\end{figure}

A closely related curve was obtained in the two-receiver prepare-and-measure network of Ref.~\cite{wang2019characterising}. The physical mechanism is the same: a nonorthogonal qubit encoding cannot be copied perfectly, so the useful information has to be split between the two receivers. The viewpoint, however, is different. In Ref.~\cite{wang2019characterising}, the characterization is formulated in terms of the Gram matrix of the preparations, so the inner products between the encoded states are part of the input data. Here, the curve appears inside the PAB model, where the broadcasting map is an explicit part of the realization and is optimized together with the preparations and measurements. Moreover, the feasible strategies are matched by SDP upper bounds, so the boundary is not only a lower-bound curve but is also certified by the relaxation used here.

This should also be distinguished from the sequential QRAC trade-off of Ref.~\cite{Mohan2019}. In the sequential case, Bob can use an unsharp measurement to extract partial information while leaving a useful post-measurement system for Charlie. In the broadcast case, the trade-off is not an information-disturbance effect in time. The qubit is distributed into two spatial shares in a single broadcasting step, and the two marginals must be compatible with the same broadcast map.

The same curve also has a natural interpretation in terms of cloning. The optimal $2\to1$ QRAC preparations can be chosen, up to a global unitary, as four equatorial states $\rho_{x_0x_1} \!=\! \frac12 \left( \iden+\vec r_{x_0x_1}\cdot\vec\sigma \right)$ with $\vec r_{x_0x_1} \!=\! \frac{1}{\sqrt2}((-1)^{x_0}, (-1)^{x_1}, 0)$.  These states form a square on a great circle of the Bloch sphere \cite{Ambainis2008QRAC}. 
Hence, in the broadcast version of the QRAC, one is not trying to distribute arbitrary qubit information, but only the equatorial information encoded by this code. This is precisely the setting of phase-covariant cloning, where a single copy of an unknown qubit restricted to the equator of the Bloch sphere is transformed into two approximate copies, each retaining as much of the original phase information as allowed by quantum mechanics~\cite{Bruss2000,DAriano2003}. In the asymmetric case, Bob's and Charlie's marginals are described by shrinking factors $\eta_B$ and $\eta_C$. For optimal asymmetric phase-covariant qubit cloners, these satisfy
\begin{equation}
    \eta_B^2+\eta_C^2\leq1,
    \label{eq:phase_covariant_shrinking_tradeoff}
\end{equation}
with equality on the optimal boundary \cite{Rezakhani2005AsymmetricPhaseCovariant,Chen2007AsymmetricPhaseCovariant}. This boundary can be realized, for instance, by the isometry $U_\alpha\ket{0}=\ket{00}$ and $U_\alpha\ket{1}=\cos\alpha\ket{10}+ \sin\alpha\ket{01}$ up to local unitaries. For equatorial inputs, this gives $\eta_B=\cos\alpha$ and $\eta_C=\sin\alpha$. Since the QRAC measurements can be chosen along the same equatorial directions, the shrinking factors simply rescale the QRAC bias: 
\begin{align}
P_B&=\frac12+\frac{\eta_B}{2\sqrt2},\\
P_C&=\frac12+\frac{\eta_C}{2\sqrt2}.
\end{align}
Therefore Eq.~\eqref{eq:phase_covariant_shrinking_tradeoff} maps exactly to Eq.~\eqref{eq:phase_covariant_qrac_curve}.

This explains why the agreement with the cloning boundary is not only a numerical coincidence. In the PAB optimization, the preparations, the isometry, and the measurements are all free. Nevertheless, the best strategies behave as if they broadcast the equatorial QRAC code through an optimal asymmetric phase-covariant cloner. In this sense, the QRAC task itself selects the relevant part of the qubit information. The trade-off is therefore not the generic no-cloning limitation for arbitrary qubit states, but the sharper limitation associated with distributing the equatorial information that gives the QRAC advantage. This is why the symmetric point remains exactly at the classical copyable value $P_B=P_C=\tfrac34$, while any genuine quantum advantage on one marginal necessarily comes at the cost of losing it on the other.

Finally, the same broadcasting scenario can be viewed in terms of effective quantum instruments. From Charlie's perspective, a fixed broadcasting isometry $U$ together with Bob's projective measurement $\qty{P^B_{b\vert y}}_b$ induces the instrument
\begin{equation}
\mathcal E_{b\vert y}(\rho)= \tr_B\!\left[(P^B_{b\vert y}\otimes\iden_C)U\rho U^\dagger(P^B_{b\vert y}\otimes\iden_C)\right].
\end{equation}
In this picture, the structure of the optimal trade-off suggests a close connection between marginal QRAC broadcasting and optimal asymmetric phase-covariant cloning transformations, in a way analogous to the connection between sequential QRACs and optimal quantum instruments~\cite{Mohan2019}.

\section{RANDOMNESS CERTIFICATION IN PREPARE-AND-MEASURE AND PREPARE-AND-BROADCAST SCENARIOS}
\label{sec:randomness}

In this section, we investigate semi-device-independent randomness certification by employing a semidefinite programming (SDP) approach inspired by the Navascués--Pironio--Acín and Navascués--Vértesi hierarchies~\cite{NPA2007,NPA2008,NV2015} to quantify the amount of certifiable randomness as a function of the prepare-and-measure or prepare-and-broadcast witnesses. Section~\ref{sec:rand-framework} sets up the adversarial model and the moment relaxation shared by all cases. The remaining subsections then apply this framework to the different scenarios considered. In Sec.~\ref{sec:pam-randomness} we treat the PAM witnesses \eqref{eq:S3} and \eqref{eq:S4}, and, while variants of the $S_4$ witness associated with quantum random access codes have been extensively studied for randomness certification, we show that the $S_3$ witness is in fact the more effective randomness witness, certifying up to one bit of randomness. Section~\ref{sec:pab-cl-randomness} turns to the broadcast witness $W^{(2)}_{\mathrm{CC}}$, whose maximal violation certifies the full two bits of joint randomness of the output pair, and Sec.~\ref{sec:qeve} strengthens the adversarial model by granting Eve a quantum system correlated with the devices. The support-function machinery common to all cases is collected in Appendix~\ref{app:randomness-support}.
 
\subsection{General framework}
\label{sec:rand-framework}
 
All certificates in this section share one adversarial model and one relaxation, which we set up once here and then apply for the different prepare-and-measure witnesses (Sec.~\ref{sec:pam-randomness}), the broadcast witness (Sec.~\ref{sec:pab-cl-randomness}), and a quantum adversary (Sec.~\ref{sec:qeve}). The only physical assumption is a bound on the dimension of the communicated system, here a qubit, while the broadcasting channel and the measurement operators remain otherwise uncharacterized. The question is whether the observed behavior alone guarantees that the generation outcomes remain unpredictable to an adversary, Eve, who may have correlated the devices through side information.
 
We first make the classical adversary precise, in the broadcast setting that contains the single-receiver case as the marginal over one output. Denote the set of quantum correlations generated in a PAB scenario by $\Qbr$,  then \begin{definition}[classical adversary]\label{rand:def:eve} A classical adversary is a random variable $\Lambda$ with distribution $\{q_\lambda\}$, statistically independent of the inputs, $q(\lambda|x,y,z)=q(\lambda)$, together with branches $p_\lambda\in\Qbr$ entering the decomposition $p=\sum_\lambda q_\lambda p_\lambda$, such that each branch specifies a complete PAB strategy, fixed before the inputs are chosen, with a qubit message, an arbitrary broadcast channel, arbitrary binary POVMs, and finite-dimensional receiver Hilbert spaces in tensor product, with no input- or output-dependent postselection.
\end{definition}
 
The single-receiver adversary of Sec.~\ref{sec:pam-randomness} is the special case in which the broadcast channel is replaced by a single qubit measurement device and $p_\lambda$ is a prepare-and-measure behavior. Refining a branch into a convex decomposition fixed before the inputs produces another adversary of the same class, a fact used repeatedly below.
 
Since $\lambda$ is known to Eve, unpredictability must be assessed branch by branch rather than at the level of the averaged behavior, and therefore the figure of merit is not the maximum output probability computed from the observed behavior. Fix a generation input $s$ and an observed witness value $W[p]\ge w$, where $W$ denotes whichever linear witness is in use. On each round Eve announces her best guess for the generation output, and given $\lambda$, her optimal guess is the most likely value under $p_\lambda$. Her guessing probability, maximized over all decompositions compatible with the observed witness value, is then
\begin{align}
G(R|\Lambda,s;w)= \sup_{\{q_\lambda,p_\lambda\}} &\sum_\lambda q_\lambda \max_{o}p_\lambda(o|s) \label{eq:guessing-general-framework}\\ \text{s.t.}\quad &p_\lambda\in\Qbr,\quad q_\lambda\geq0,\quad\sum_\lambda q_\lambda=1, \nonumber\\ &\sum_\lambda q_\lambda W[p_\lambda]\geq w, \nonumber
\end{align}
where $o\in \mathcal O_R$ denotes the outcome associated with $R \in\{BC,B,C\}$. The certified min-entropy is then $H_{\min}(R|\Lambda,s)=-\log_2 G(R|\Lambda,s;w)$. 
 
Two obstacles prevent a direct semidefinite relaxation of Eq.~\eqref{eq:guessing-general-framework}. The number of values of $\lambda$ can in principle be unbounded, and the inner maximization over the outcome is nonlinear in the decomposition. Both are removed by grouping the branches according to Eve's guess. Writing $\mathcal G$ for the set of possible guesses $g$ and $g_\lambda\in\mathcal G$ for an optimal guess of branch $\lambda$, $g_\lambda\in \text{argmax}_{o\in\mathcal{O}_R}\; p_\lambda(o|s)$, the subnormalized behaviors
\begin{equation}
\widetilde p_g(o'|s') =\sum_{\lambda:g_\lambda=g}q_\lambda\,p_\lambda(o'|s') 
\label{eq:subnormalized-framework}
\end{equation}
carry the weights $r_g=\sum_{\lambda:g_\lambda=g}q_\lambda$, with $\sum_g r_g=1$. Because the objective function evaluates each branch only at its own guess and $W$ is linear, both the guessing probability and the witness value of any decomposition are reproduced exactly by the grouped form, so, with $\operatorname{cone}(\Qbr)$ denoting the set of nonnegative multiples of quantum behaviors, we obtain 
\begin{align}
G(R|\Lambda,s;w)= \sup_{\{\widetilde p_g\}} &\sum_{g\in\mathcal G}\widetilde p_g(g|s) \label{eq:guessing-cone-framework}\\ \nonumber \text{s.t.}\quad &\widetilde p_g\in\operatorname{cone}(\Qbr), \quad \sum_g r_g=1, \\ &  \sum_g W[\widetilde p_g]\geq w. \nonumber
\end{align}
The reduction is lossless, since any grouped decomposition arises from an adversarial strategy and conversely. Marginal randomness of a single receiver uses $\mathcal G=\{0,1\}$ and two branches, joint randomness of the pair uses $\mathcal G=\{0,1\}^2$ and four.  
 
The set $\Qbr$ is not itself semidefinite representable, so we relax each branch by a finite-dimensional moment relaxation, inspired by the Navascu\'es--Pironio--Ac\'in and Navascu\'es--V\'ertesi hierarchies~\cite{NPA2007,NPA2008,NV2015} but parametrizing the preparation amplitudes explicitly and imposing the lifted moment constraints directly, with independent truncation levels for preparations and measurement words. We label its levels $L(q,t)$, with $q$ the scalar degree retained for the preparation amplitudes and $t$ the maximal length of the measurement words, and we write $L\in\mathcal S_{q,t}(r)$ for a moment functional whose lifted moment matrix is positive semidefinite, satisfies the linear identities of the preparation normalization and the binary measurement algebra, and whose associated behavior has total weight $r$ at every input, the case $r=1$ recovering the normalized relaxation. Both the probabilities and the witness value are affine in the moments, so replacing each branch of Eq.~\eqref{eq:guessing-cone-framework} by a functional $L_g\in\mathcal S_{q,t}(r_g)$, can be cast as a semidefinite program:
\begin{align}
G_{q,t}(R|\Lambda,s;w)= \max_{\{L_g\}} &\sum_{g\in\mathcal G}p_{L_g}(g|s) \label{eq:guessing-sdp-framework}\\ \nonumber \text{s.t.}\quad &L_g\in\mathcal S_{q,t}(r_g), \quad \sum_g r_g=1, \\ & \sum_g W[L_g]\geq w, \nonumber
\end{align}
which can only increase the optimum. Hence $G\le G_{q,t}$ and the certified min-entropy $H_{\min}(R|\Lambda,s)\ge-\log_2 G_{q,t}(R|\Lambda,s;w)$, a bound that is nontrivial whenever the observed witness value is incompatible with any decomposition in which Eve predicts the generation outcome perfectly. Solving Eq.~\eqref{eq:guessing-sdp-framework} pointwise in $w$ is possible but wasteful. In practice we reconstruct the whole curve at once from the support functions of the individual branches, a construction we carry out in Appendix~\ref{app:support_function_lower_bounds} and use throughout this section. The resulting envelope is a conservative upper bound on Eve's guessing probability, valid for every observed witness value and computed from single-branch semidefinite programs.
 
\subsection{Randomness certification in the PAM scenario with classical side information}
\label{sec:pam-randomness}
We apply now the framework of Sec.~\ref{sec:rand-framework} for a single receiver, so the adversary of Definition~\ref{rand:def:eve} reduces to a classical variable $\lambda$ selecting a qubit prepare-and-measure branch $p_\lambda$, the guess set is $\mathcal G=\{0,1\}$, and the certified quantity is the min-entropy of Bob's single outcome $b$ at a fixed generation input $s_{\star}=(x^{\star},y^{\star})$, bounded above by one bit. 
 
Our certificates are built on the witnesses $S_3$ and $S_4$ introduced in Sec.~\ref{sec:thepampab}, whose classical bounds are $S_3^{\mathrm C}=3$ and $S_4^{\mathrm C}=4$ and whose maximal qubit values are $S_3^{\mathrm Q}=1+2\sqrt{2}$ and $S_4^{\mathrm Q}=4\sqrt{2}$. Throughout this subsection, the generation input is fixed to $(x^{\star},y^{\star})=(2,1)$ for two reasons: all inputs choices lead to the same results for $S_4$, and $(x^{\star},y^{\star})=(2,1)$ is the optimal input at which the qubit realization attaining $S_3^{\mathrm Q}$ produces an unbiased outcome, so that one full bit is in principle available. We compare two uses of the observed data, one in which Eve is constrained only by the value of the relevant witness, and one in which she is constrained by the complete observed behavior. Since the latter retains at least as much information as the scalar witness value, it can only maintain or strengthen the randomness certificate.
 
For the witness-only certificate, specializing Eq.~\eqref{eq:guessing-cone-framework} to two branches $\widetilde p_g(b|x,y)$, $g\in\{0,1\}$, of weights $r_g=\sum_b\widetilde p_g(b|x,y)$, gives
\begin{align}
G^{\mathrm{wit}}_{k}(B|\Lambda,s_{\star};s) = \sup_{\{\widetilde p_g\}} \quad& \widetilde p_0(0|s_{\star}) + \widetilde p_1(1|s_{\star}) \nonumber\\ \mathrm{s.t.}\quad & \widetilde p_g \in \operatorname{cone} \bigl(\mathcal Q^{\mathrm{PAM}}_{2}\bigr), \nonumber\\ & r_0+r_1=1, \:\:\: \sum_{g=0}^{1}S_k[\widetilde p_g]\geq s,
\label{eq:pam-witness-two-branch}
\end{align}
where $\mathcal Q^{\mathrm{PAM}}_{2}$ denotes the set of behaviors generated by qubit communication and $S_k[p]\geq s$ is the observed witness constraint, $k\in\{3,4\}$.
 
A witness value compresses the experimental data into a single scalar, and behaviors sharing the same value of $S_k$ may nevertheless allow different guessing probabilities. To quantify what is gained by keeping the full statistics, let $p^{(k)}_{\mathrm{opt}}(b|x,y)$ be a qubit behavior attaining $S_k^{\mathrm Q}$, and consider the noisy family
\begin{equation}
p^{(k)}_{v}(b|x,y) = v\,p^{(k)}_{\mathrm{opt}}(b|x,y) + (1-v)\frac{1}{2},
\label{eq:pam-noisy-full-behavior}
\end{equation}
where $0\leq v\leq1$, for which $S_k[p^{(k)}_{v}]=vS_k^{\mathrm Q}$, since white noise has vanishing correlators. Constraining Eve to reproduce the complete behavior leads to
\begin{align}
G^{\mathrm{dist}}_{k}(B|\Lambda,s_{\star};v) = \sup_{\{\widetilde p_g\}} \quad& \widetilde p_0(0|s_{\star}) + \widetilde p_1(1|s_{\star}) \nonumber\\ \mathrm{s.t.}\quad & \widetilde p_g \in \operatorname{cone} \bigl(\mathcal Q^{\mathrm{PAM}}_{2}\bigr), \nonumber\\ & \sum_{g=0}^{1} \widetilde p_g(b|x,y) = p^{(k)}_{v}(b|x,y) \nonumber\\ & \hspace{7.5em} \forall\,b,x,y.
\label{eq:pam-full-behavior-program}
\end{align}
The complete-behavior constraint implies the corresponding witness constraint, while the converse fails in general, so 
\begin{equation}
\begin{split}
G^{\mathrm{dist}}_{k}(B|\Lambda,s_{\star};v) \leq G^{\mathrm{wit}}_{k}\bigl(B|\Lambda,s_{\star};vS_k^{\mathrm Q}\bigr),
\end{split}
\label{eq:pam-dist-vs-wit}
\end{equation}
and hence $H^{\mathrm{dist}}_{\min,k}\geq H^{\mathrm{wit}}_{\min,k}$. The gap between the two curves measures the randomness advantage obtained by retaining the complete statistics rather than compressing them into a single witness value.
 
We relax each subnormalized branch at the level $L(1,2)$ of the moment relaxation of Sec.~\ref{sec:rand-framework}. Here the measurement operators are represented by the word set $\mathcal S_2=\{\id,B_0,B_1,B_0B_1,B_1B_0\}$ with $B_y^\dagger=B_y$ and $B_y^2=\id$, the communicated dimension is fixed to $d=2$, and the common unitary freedom of the qubit input space allows the gauge $\ket{\psi_0}=|0\rangle$ without restricting the set of physical behaviors. Each guess $g\in\{0,1\}$ is assigned an independent moment matrix $\Gamma_g$, whose rows and columns are indexed by preparation monomials $m,n$, input-basis labels $i,j\in\{0,1\}$, and measurement words $u,v\in\mathcal S_2$, with entries
\begin{equation}
\Gamma_g \bigl[(m,i,u),(n,j,v)\bigr] = L_g\!\left( m^\ast n\, \langle i|u^\dagger v|j\rangle \right),
\label{eq:pam-randomness-moment-matrix}
\end{equation}
where $\Gamma_g\succeq0$ and $L_g$ denotes the subnormalized moment functional associated with the branch in which Eve guesses $g$, and the preparation-normalization conditions together with the measurement identities enter as linear relations among the entries. Subnormalization is encoded in the zeroth-order block, $\Gamma_g\bigl[(1,i,\id),(1,j,\id)\bigr]$ where $r_g\geq0$, and $r_0+r_1=1$, so that the identity block is normalized to the total weight of the corresponding adversarial branch rather than to unity. Both the probabilities $p_{L_g}(b|x,y)$ and the witness values $S_k[L_g]$ are affine in the entries of $\Gamma_g$, and at the chosen level the $S_3$ and $S_4$ relaxations involve moment matrices of order $90$ and $130$, respectively, for each of Eve's two guesses. Every physical qubit PAM strategy generates a feasible moment matrix, whereas the converse need not hold at any finite level, so the relaxation yields the upper bound $G_k\leq G_k^{L(1,2)}$ and the certified min-entropy
\begin{equation}
H_{\min,k}(B|\Lambda,s_{\star})
\geq
-\log_2
G_k^{L(1,2)}(B|\Lambda,s_{\star}).
\label{eq:pam-L12-minentropy}
\end{equation}
For the witness-only curves, this bound is reconstructed from an adaptively refined family of support-function semidefinite programs, while the fixed-distribution curves are obtained by optimizing the two subnormalized moment matrices simultaneously under the complete-behavior constraints of Eq.~\eqref{eq:pam-full-behavior-program}. All reported curves follow directly from these semidefinite programs.  

\begin{figure}
    \centering
    \includegraphics{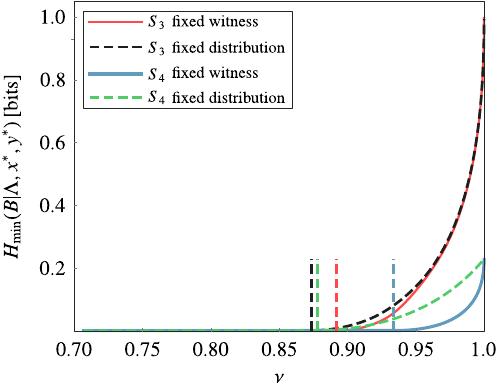}
 \caption{\textbf{Randomness from the $S_3$ and $S_4$ witnesses in the PAM scenario.} Certified min-entropy of Bob's outcome at the generation input $(x^\star,y^\star)=(2,1)$ as a function of the visibility $v$ for the witnesses $S_3$ and $S_4$, at the $L(1,2)$ relaxation level. For each witness, we show the witness-only certificate of Eq.~\eqref{eq:pam-witness-two-branch} and the fixed-distribution certificate of Eq.~\eqref{eq:pam-full-behavior-program}. For $S_3$, the fixed-distribution certificate reaches one full bit at $v=1$, whereas for $S_4$ the bias of the optimal realization at the same generation input limits the certifiable randomness to approximately $0.228$ bits. The vertical lines mark the minimum visibility required for each curve to certify nonzero randomness.}
    \label{fig:pam-randomness}
\end{figure}
 
Figure~\ref{fig:pam-randomness} collects the resulting certificates as functions of the visibility $v$. Two features stand out. First, for $S_3$, the fixed-distribution certificate reaches one full bit at $v=1$, the algebraic maximum for a binary outcome, in agreement with the unbiased marginal of the optimal realization at the input $(2,1)$. Second, away from the endpoint, the two protocols separate, with the fixed-distribution curve lying above the witness-only envelope over a wide range of visibilities and becoming nontrivial at a lower visibility threshold. The witness value is therefore a valid but strictly lossy summary of the data for randomness purposes. The witness $S_4$ shows the complementary lesson. At the same generation input, the realization attaining $S_4^{\mathrm Q}$ is biased, with $E_{21}=\tfrac{1}{\sqrt{2}}$, which limits the certifiable randomness at $v=1$ to $-\log_2[\tfrac12(1+\tfrac{1}{\sqrt{2}})]\simeq0.228$ bits. A witness with more preparations is not automatically more useful for randomness at a given input, and the choice of generation input matters as much as the choice of witness.

\subsection{Randomness certification in the PAB scenario with classical side information}
\label{sec:pab-cl-randomness}
 
In device-independent Bell experiments, nonlocal correlations certify randomness without requiring trust in the internal functioning of the measurement devices~\cite{Pironio2010}. For two binary outputs, the joint output contains at most two bits of randomness. Achieving this maximum, however, is highly nontrivial in the standard bipartite Bell scenario~\cite{Acin2016,Wooltorton2022}. Tight analytic trade-offs between the CHSH value and global randomness have been derived~\cite{Wooltorton2022}, showing that two bits of global randomness can be certified only by tailored families of Bell expressions, and only for CHSH values in the interval $(2,3\sqrt3/2]$, while at the Tsirelson value of the CHSH inequality, one can certify at most $-\log_2[(1+1/\sqrt2)/4]\simeq1.228$ bits of min-entropy for the joint output pair. Complementary analyses have shown that using the full observed distribution, rather than compressing the data into a single Bell value, can strictly improve the certified randomness bound~\cite{NietoSilleras2014,Bancal2014}. These results motivate two principles that structure this section, and that Sec.~\ref{sec:pam-randomness} has already displayed in the single-receiver setting. The witness that is optimal for detecting nonclassicality need not be optimal for certifying randomness, and a scalar witness value may be a suboptimal summary of the experimental data. In particular, we will analytically prove that the PAB witness $W_{\mathcal{PAB}}$ of Eq.~\eqref{eq:CC2}, introduced in Ref.~\cite{sarubi2026prepare}, enables the certification of the full two bits of joint randomness within the prepare-and-broadcast scenario.
  
We apply the framework of Sec.~\ref{sec:rand-framework} for the two receivers, so the adversary is that of Definition~\ref{rand:def:eve}, the guess set is $\mathcal{G}_{BC}=\{(0,0),(0,1),(1,0),(1,1)\}$, and the certified quantity is the joint min-entropy of the output pair $(B,C)$. No coherent quantum memory is granted to Eve in this subsection, a restriction lifted in Sec.~\ref{sec:qeve}. The witness is $W_{\mathcal{PAB}}$ of Eq.~\eqref{eq:CC2}, hereafter denoted $W$, with classical bound $W_C=6$ and maximal value $W_Q=8+2\sqrt2\simeq10.83$ over qubit PAB strategies. At the maximal violation, and for suitable generation inputs, the certificate reaches the two-bit endpoint in the following operational sense. The explicit qubit strategy produces a uniform output pair at the generation input, while the semidefinite upper bound on Eve's guessing probability gives $G(BC|\Lambda,s;W_Q)\leq\tfrac14$ up to the reported numerical precision, so the certified min-entropy is two bits within this classical-side-information model. Theorem~\ref{rand:thm:main} below turns this numerical endpoint into an exact analytic statement.
 
The physical set $\Qbr$ entering the certification consists of the qubit PAB behaviors
\begin{equation}
p(b,c|x,y,z) =\tr\!\left[\Lambda(\rho_x)\left(B_{b|y}\otimes C_{c|z}\right)\right], 
\label{eq:pab-physical-set}
\end{equation}
with $\rho_x\in\mathcal{D}(\mathbb{C}^2)$, a broadcast channel $\Lambda$, and local POVMs for Bob and Charlie. After Stinespring and Naimark dilations~\cite{Stinespring1955,Naimark1940}, every such behavior admits the exact projective rewriting
\begin{equation}
p(b,c|x,y,z) =\langle\psi_x|V^\dagger B_{b|y}C_{c|z}V|\psi_x\rangle, 
\label{eq:pab-dilated-form}
\end{equation}
with $\ket{\psi_x}\in\mathbb{C}^2$, an isometry $V:\mathbb{C}^2\to\mathcal{K}$, and Bob's and Charlie's projectors commuting on the dilation space. This dilated form is the representation relaxed by the moment construction of Sec.~\ref{sec:rand-framework}, now built on the binary broadcast algebra $B_y^2=C_z^2=\id$, $[B_y,C_z]=0$. The guessing problem Eqs.~\eqref{eq:guessing-general-framework}--\eqref{eq:guessing-sdp-framework} then applies verbatim with $\max_{b,c}$ in place of the generic outcome maximization and $W$ the linear functional of Eq.~\eqref{eq:CC2}, giving the four-branch semidefinite program $G_{q,t}(BC|\Lambda,s;w)$ and the certified bound $H_{\min}(BC|\Lambda,s)\ge-\log_2 G_{q,t}(BC|\Lambda,s;w)$. The marginal randomness of a single receiver follows with two branches, the objective using $p_{L_g}(g|x^*,y^*)=\sum_c p_{L_g}(g,c|x^*,y^*,z^*)$, independent of $z^*$ for nonsignaling behaviors.
 
Figure~\ref{fig:pab-all-inputs} shows the resulting certificates for every generation input directly as functions of the observed witness value $W_{\mathcal{PAB}}$. The certified randomness depends strongly on the input used for generation. For suitable inputs, among them $s=(1,0,0)$, the joint certificate grows monotonically with $W_{\mathcal{PAB}}$  and approaches two bits as the witness approaches its maximal quantum value $W_Q$, up to the reported numerical precision. Other inputs saturate at or below one bit, in agreement with the correlations imposed by the optimal realization at those inputs. As in the single-receiver case, we also evaluate the fixed-distribution protocol along the white-noise family $p_v=v\,p_{\mathrm{opt}}+(1-v)p_{\mathrm{unif}}$. Since the uniform behavior has vanishing correlators, $W_{\mathcal{PAB}}[p_v]=vW_Q$. Thus, although the physical noisy family is parametrized by the visibility $v$, the horizontal coordinate in Fig.~\ref{fig:pab-all-inputs} is the corresponding observed witness value $W_{\mathcal{PAB}}=vW_Q$.
 
\begin{figure}[t]
    \centering
    \includegraphics{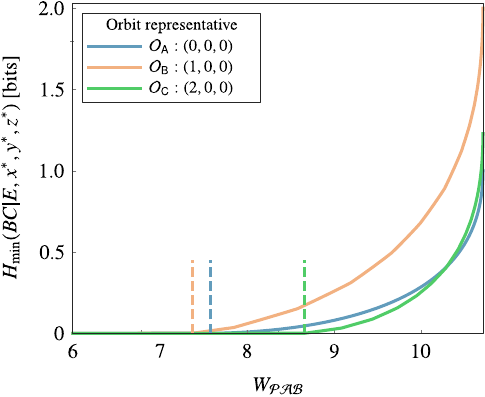}
    \caption{\textbf{Joint randomness for every generation input in the PAB scenario.}Joint randomness certification from the witness $W_{\mathcal{PAB}}$ of Eq.~\eqref{eq:CC2} for all twelve generation inputs $(x,y,z)$. The curves show the certified min-entropy of the output pair $(B,C)$ as a function of the observed witness value $W_{\mathrm{PAB}}$, reconstructed from the support-function envelope of Eq.~\eqref{eq:support-envelope} at the $L(1,2)$ level. For suitable generation inputs, the certificate approaches two bits as $W_{\mathrm{PAB}}$ approaches its maximal quantum value $W_Q=8+2\sqrt2$, the algebraic maximum for a pair of binary outputs. Under the symmetry of Corollary~\ref{rand:cor:orbit}, the twelve curves collapse onto three, one per orbit $\mathcal O_A$, $\mathcal O_B$, and $\mathcal O_C$, and only the orbit $\mathcal O_B$, which contains the generation input $s^\star=(1,0,0)$ of Theorem~\ref{rand:thm:main}, reaches two bits.}
    \label{fig:pab-all-inputs}
\end{figure}
 
The endpoint visible in Fig.~\ref{fig:pab-all-inputs} is not only numerical. For the generation input $s^*=(1,0,0)$ it is exact, and the symmetry of the witness extends the statement to an entire orbit of inputs.
 
\begin{theorem}[two bits at the maximal violation]\label{rand:thm:main}
Let $s^*=(1,0,0)$ and let $G$ be the guessing probability \eqref{eq:guessing-general-framework} against the adversaries of Definition~\ref{rand:def:eve}. Then
\begin{align}\label{rand:eq:endpoint}
G(BC|\Lambda,s^*;\WQ)&=\tfrac14, \\
\Hmin(BC|\Lambda,s^*)&=2 ,
\end{align}
both exactly. Moreover, every adversarial branch of nonzero weight produces the uniform distribution on the output pair at $s^*$. 
\end{theorem}
 
\begin{corollary}[the four-input orbit]\label{rand:cor:orbit}
The relabelings of the PAB experiment that leave the coefficient tensor of $W_{\mathcal{PAB}}$ invariant form a group of order $16$, whose action partitions the twelve generation inputs into three orbits of size four, 
\begin{equation}\label{rand:eq:orbits}
\begin{aligned}
O_A&=\{(0,0,0),(0,1,1),(1,0,1),(1,1,0)\},\\
O_B&=\{(0,0,1),(0,1,0),(1,0,0),(1,1,1)\},\\
O_C&=\{(2,0,0),(2,0,1),(2,1,0),(2,1,1)\}.
\end{aligned}
\end{equation}
Guessing-probability curves coincide along each orbit, $G(BC|\Lambda,s;w)=G(BC|\Lambda,s';w)$ for $s,s'$ in the same orbit and every $w\le\WQ$, so Theorem~\ref{rand:thm:main} holds verbatim for every $s\in O_B$. At the maximal violation the inputs in $O_A$ instead display perfect correlation of the output pair on every branch, so at most one bit of joint randomness is available there, while for $O_C$ the certification remains numerical. 
\end{corollary}
 
The proof is given in Appendix~\ref{app:endpoint}, and every step is exact. Stinespring and Naimark dilations reproduce the observed statistics exactly, so every branch may be taken in the projective isometric form of Eq.~\eqref{eq:pab-dilated-form}, with the three dilated preparations confined to the two-dimensional code subspace. The maximal violation saturates the three preparation blocks of $W_{\mathcal{PAB}}$ separately, which a sum-of-squares identity converts into algebraic relations among the dilated observables, and a Jordan-block analysis transfers the vanishing of two anticommutator defects through the rank-two constraint imposed by the qubit message. The dimension bound thus plays the structural role that entanglement rigidity plays in Bell self-testing~\cite{vsupic2020self}. A structural reason for the two-bit value is visible already in Eq.~\eqref{eq:CC2}, the generation correlator $\langle B_0C_0\rangle^{[1]}$ does not appear in $W_{\mathcal{PAB}}$, so the witness can be saturated while the generation statistics remain completely unbiased, and the saturation of the three blocks forces exactly that.
 
A structural feature of the maximal violation deserves emphasis before any comparison is drawn. At $W_{\mathcal{PAB}}=W_{\mathcal{PAB}}^{\mathrm Q}$, the conditional behavior at the generation input is not merely unbiased but Bell local. Corollary~\ref{rand:cor:local} shows that all four single-party marginals vanish and that, at $x=1$, one has $\corr{0}{1}{1}=1$, $\corr{1}{0}{1}=-1$, and $\corr{0}{0}{1}=\corr{1}{1}{1}=0$. Every CHSH relabeling therefore has modulus at most the local bound $2$, and the slice is reproduced by an explicit local hidden-variable model. Thus, the full two bits of joint randomness are certified at an input whose conditional Bell behavior is local; the certification comes instead from the qubit dimension constraint and from the correlations involving the other preparation inputs.

Figure~\ref{fig:pab-w2-vs-chsh} should therefore be read as a comparison between two separately optimized white-noise benchmarks, not as two analyses of the same data. The PAB curves are generated from
\begin{equation}
p_v^{W}
=
v\,p_{\mathrm{opt}}^{W}
+
(1-v)p_{\mathrm{unif}},
\end{equation}
where $p_{\mathrm{opt}}^{W}$ attains $W_{\mathcal{PAB}}^{\mathrm Q}$, whereas the CHSH curves use
\begin{equation}
p_v^{\mathrm{CHSH}}
=
v\,p_{\mathrm{opt}}^{\mathrm{CHSH}}
+
(1-v)p_{\mathrm{unif}},
\end{equation}
with $p_{\mathrm{opt}}^{\mathrm{CHSH}}$ maximally violating CHSH. This distinction is necessary because CHSH remains within its local bound along the maximal-$W_{\mathcal{PAB}}$ family and hence gives no CHSH-based certificate there. In both cases the observed statistic scales linearly with $v$,
\begin{equation}
W_{\mathcal{PAB}}[p_v^{W}]
=
v\,W_{\mathcal{PAB}}^{\mathrm Q},
\qquad
\mathrm{CHSH}[p_v^{\mathrm{CHSH}}]
=
v\,\mathrm{CHSH}_{\mathrm Q},
\end{equation}
so the common horizontal axis represents the same normalized white-noise mixing parameter, although not the same experiment. The PAB certificates are computed against classical side information with the qubit-bounded $L(1,2)$ relaxation of Sec.~\ref{sec:rand-framework}; the CHSH curves use the standard dimension-unbounded NPA relaxation with operator words of length at most two~\cite{Acin2016,Wooltorton2022}. In the fixed-distribution variant, Eve is required to reproduce the full behavior defining each family: $p_v^{W}(b,c|x,y,z)$ for all three preparation inputs in the PAB case, and the complete Bell behavior $p_v^{\mathrm{CHSH}}(b,c|y,z)$ in the CHSH case.

\begin{figure}[t]
    \centering
    \includegraphics{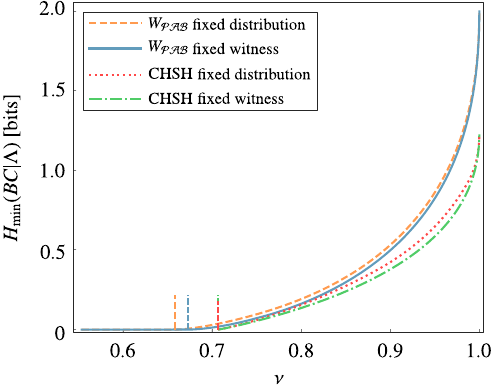}
    \caption{\textbf{Broadcast witness against the CHSH benchmark.}
    Certified joint min-entropy versus the normalized white-noise mixing
    parameter $v$ for two separately optimized benchmark families. The
    $W_{\mathcal{PAB}}$ curves are evaluated at the qubit-bounded $L(1,2)$
    relaxation, while the CHSH curves use the dimension-unbounded NPA
    relaxation with operator words of length at most two. For each benchmark
    we show a witness-only certificate and a fixed-distribution certificate.
    The vertical lines mark the onset of the certificates, at
    $v\simeq0.659$ and $v\simeq0.674$ for $W_{\mathcal{PAB}}$, and near
    $v=1/\sqrt2\simeq0.707$ for CHSH. At the relaxation levels considered,
    the PAB curves remain above their corresponding CHSH benchmarks
    throughout the CHSH-nonlocal range and reach two bits at $v=1$, compared
    with approximately $1.228$ bits for maximal CHSH violation.}
    \label{fig:pab-w2-vs-chsh}
\end{figure}

The advantage of the PAB construction is then visible both at threshold and throughout the noisy regime. Its fixed-distribution and witness-only certificates become nontrivial at $v\simeq0.659$ and $v\simeq0.674$, respectively, both below the CHSH locality threshold $1/\sqrt2\simeq0.707$. Moreover, at the relaxation levels considered here, each PAB curve remains above its CHSH counterpart over the entire visibility range in which the CHSH family is nonlocal. The gap persists up to the ideal point, where $W_{\mathcal{PAB}}$ certifies two bits, compared with approximately $1.228$ bits for maximal CHSH violation~\cite{Wooltorton2022}. This enhanced robustness reflects the fact that $W_{\mathcal{PAB}}$ combines information from all three preparation inputs, while CHSH probes only a single conditional Bell behavior; in addition, maximal CHSH violation fixes the magnitude of every tested correlator to $1/\sqrt2$, whereas the correlator at the PAB generation setting does not enter the witness. For both benchmarks, retaining the complete noisy behavior yields a stronger certified entropy bound than retaining only the scalar witness value.

\subsection{Quantum Eve}
\label{sec:qeve}
 
The preceding analyses restrict Eve to classical side information, represented by a variable $\lambda$ that selects the strategy implemented in each round. We now strengthen the adversarial model and allow Eve to retain a quantum system correlated with Bob's device. To isolate the effect of this stronger form of side information, we work with the three-preparation marginal of the PAB scenario, characterized by the witness $S_3$. In the broadcasting architecture, the second output of the channel would be delivered to Charlie. Here Eve controls the channel and keeps that entire complementary output, while only Bob's marginal behavior is observed. The structural difference from the classical analysis is that Eve's information is now retained coherently, so the guessing problem no longer splits into guess-labeled subnormalized branches, as made precise below.
 
Alice prepares a qubit state $|\psi_x\rangle\in\mathbb C^2$, and the adversarial channel is represented by an isometry
\begin{equation}
U:\mathbb C^2\longrightarrow \mathcal H_B\otimes\mathcal H_E.
\label{eq:qeve-isometry}
\end{equation}
Bob receives the subsystem in $\mathcal H_B$ and performs a binary measurement labeled by $y$, while Eve stores the subsystem in $\mathcal H_E$. Once the generation input $s_\star=(x^\star,y^\star)$ is publicly revealed, Eve performs the binary measurement that maximizes her probability of predicting Bob's outcome. Since each optimization fixes one generation input, a single dichotomic observable $E_\star$ suffices to represent her measurement. The isometry in Eq.~\eqref{eq:qeve-isometry} describes Eve's control of the channel and carries no additional purification assumption on Alice's source. The input system remains a qubit, while the output dimensions of Bob and Eve are unrestricted. After a Naimark dilation, Bob's and Eve's binary measurements may be represented by Hermitian observables satisfying
\begin{equation}
\begin{gathered}
B_y^\dagger=B_y,
\qquad
E_\star^\dagger=E_\star,\\
B_y^2=E_\star^2=\id,
\qquad
[B_y,E_\star]=0,
\end{gathered}
\label{eq:qeve-measurement-algebra}
\end{equation}
the commutation relation encoding that Bob and Eve act on distinct outputs of the isometry.
 
Writing $\langle O\rangle_x=\langle\psi_x|U^\dagger O U|\psi_x\rangle$ for the expectation values at a fixed preparation $x$, the joint Bob--Eve statistics read
\begin{align}
p(b,e|x,y;s_\star) = \frac{1}{4} \Big[ 1 &+(-1)^b\langle B_y\rangle_x +(-1)^e\langle E_\star\rangle_x \nonumber\\ &+(-1)^{b+e}\langle B_yE_\star\rangle_x \Big],
\label{eq:qeve-joint-probabilities}
\end{align}
and marginalizing over Eve returns the observed PAM behavior,
\begin{equation}
p_B(b|x,y) = \sum_e p(b,e|x,y;s_\star) = \frac{1}{2} \left[ 1+(-1)^b\langle B_y\rangle_x \right].
\label{eq:qeve-bob-marginal}
\end{equation}
The observed witness is therefore unchanged,
\begin{equation}
S_3 = \langle B_0\rangle_0 +\langle B_1\rangle_0 +\langle B_0\rangle_1 -\langle B_1\rangle_1 -\langle B_0\rangle_2,
\label{eq:qeve-S3}
\end{equation}
with classical bound $S_3^{\mathrm C}=3$ and maximal qubit value $S_3^{\mathrm Q}=1+2\sqrt2$. 
 
At the generation input, Eve succeeds whenever her outcome coincides with Bob's. Her quantum guessing probability is therefore
\begin{align}
G^{\mathrm Q}(B|E,s_\star;s)
=
\sup\quad&
\sum_{b=0}^{1}
p(b,e=b|x^\star,y^\star;s_\star)
\nonumber\\
\mathrm{s.t.}\quad&
S_3[p_B]\geq s,
\nonumber\\
&
\dim\mathcal H_{\mathrm{in}}=2,
\quad
\text{Eqs.~\eqref{eq:qeve-isometry}--\eqref{eq:qeve-measurement-algebra}},
\label{eq:qeve-guessing-problem}
\end{align}
whose objective reduces, in terms of dichotomic observables, to
\begin{equation}
G^{\mathrm Q}(B|E,s_\star;s)
=
\sup\,
\frac{1}{2}
\left[
1+
\langle B_{y^\star}E_\star\rangle_{x^\star}
\right],
\label{eq:qeve-guessing-correlator}
\end{equation}
with conditional min-entropy
\begin{equation}
H_{\min}(B|E,s_\star)
=
-\log_2 G^{\mathrm Q}(B|E,s_\star).
\label{eq:qeve-minentropy}
\end{equation}
In contrast with the classical-side-information problem, the optimization in Eq.~\eqref{eq:qeve-guessing-problem} does not split into subnormalized branches labeled by a predetermined guess. A guess exists only once Eve's quantum system is measured, after the generation input is revealed, so the entire optimization is carried by a single normalized Bob--Eve moment matrix \cite{sarubi2025pab}.
 
\begin{figure*}[!t]
    \centering
    \includegraphics{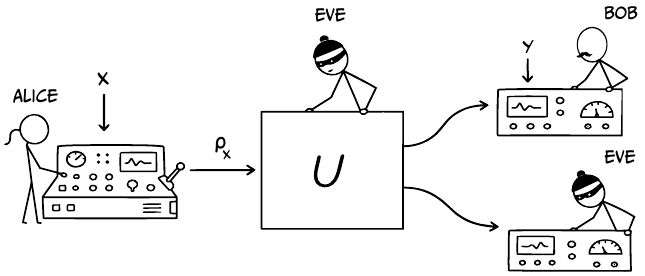}
    \caption{\textbf{Quantum-side-information architecture for the marginal
    prepare-and-measure benchmark}. Alice prepares a qubit state
    $\rho_x$, and Eve implements an arbitrary isometry
    $U:\mathbb{C}^2\to\mathcal{H}_B\otimes\mathcal{H}_E$. Bob receives one
    output and performs the measurement $B_y$, while Eve retains the entire
    complementary output and measures it only after the fixed generation input
    $(x^*,y^*)$ is revealed. In the original broadcasting interpretation, the
    second output is the arm that would otherwise be assigned to Charlie.}
    \label{fig:quantum-broadcast}
\end{figure*}
 
We bound Eq.~\eqref{eq:qeve-guessing-problem} with the same finite-dimensional moment relaxation used in the preceding subsections, now applied to the commuting Bob--Eve operator algebra. For measurement words $u,v$, define the frame moments 
\begin{equation}
\eta^{ij}_{u,v} = \langle i| U^\dagger u^\dagger v U |j\rangle, \qquad i,j\in\{0,1\},
\label{eq:qeve-frame-moments}
\end{equation}
into which the isometry is absorbed, so that it never needs to be explicitly parametrized. With preparation monomials $m,n$ retained up to
scalar degree $q=1$, the moment matrix has entries
\begin{equation}
\Gamma^{(\ell)} \bigl[(m,i,u),(n,j,v)\bigr] = L^{(\ell)} \left( m^\ast n\,\eta^{ij}_{u,v} \right),
\label{eq:qeve-moment-matrix}
\end{equation}
where $\Gamma^{(\ell)}\succeq0$. Preparation normalization, the qubit input dimension, and the operator relations of Eq.~\eqref{eq:qeve-measurement-algebra} enter as linear identities among the entries of $\Gamma^{(\ell)}$, and the unitary freedom of the input qubit again allows the gauge $|\psi_0\rangle=|0\rangle$ without restricting the physical set.
 
We compare two nested operator levels. The augmented first level, denoted $L(1,1)+BE$, uses the word set
\begin{equation}
\mathcal W_{1}^{BE} = \{ \id, B_0,B_1, E_\star, B_0E_\star, B_1E_\star \},
\label{eq:qeve-level-one-words}
\end{equation}
where the mixed words $B_yE_\star$ must be included because the
objective in Eq.~\eqref{eq:qeve-guessing-correlator} depends directly on
the Bob--Eve correlator. The second level, denoted $L(1,2)$, contains
every reduced commuting word of total length at most two,
\begin{equation}
\mathcal W_{2}^{BE} = \mathcal W_{1}^{BE} \cup \{ B_0B_1,B_1B_0 \}.
\label{eq:qeve-level-two-words}
\end{equation}
After the input gauge is fixed, these levels produce moment matrices of order $132$ and $176$, respectively, for each fixed generation input. Since every physical Bob--Eve realization generates feasible moments, each level is an outer approximation of the physical set and yields, for $\ell\in\{L(1,1)+BE,\,L(1,2)\}$, 
\begin{equation}
G^{\mathrm Q}(B|E,s_\star;s) \leq G_{\ell}^{\mathrm Q}(B|E,s_\star;s),
\label{eq:qeve-relaxed-bound}
\end{equation}
with certified entropy
\begin{equation}
H_{\min}^{\mathrm Q,\ell}(B|E,s_\star;s) = -\log_2 G_{\ell}^{\mathrm Q}(B|E,s_\star;s).
\label{eq:qeve-relaxed-entropy}
\end{equation}
Since the second word set contains the first, the $L(1,2)$ relaxation is formally at least as tight as the augmented first level. 
 
As in the classical-side-information analysis, we parametrize the witness target by the visibility $v=s/S_3^{\mathrm Q}$, or equivalently $s=vS_3^{\mathrm Q}$. Along the white-noise family $p_v=v\,p_{\mathrm{opt}}+(1-v)p_{\mathrm{unif}}$, this parameter coincides with the physical mixing visibility because $S_3[p_v]=vS_3^{\mathrm Q}$. We reconstruct the certified curve $G_{\ell}^{\mathrm Q}(B|E,s_\star;v)$ from the centered single-branch support function of Appendix~\ref{app:randomness-support}, Eqs.~\eqref{eq:qeve-centered-support}-- \eqref{eq:qeve-support-envelope}.
 
To compare quantum and classical side information on an equal footing, we evaluate three certificates, the quantum Eve at level $L(1,1)+BE$, the quantum Eve at level $L(1,2)$, and the classical side information at the PAM level $L(1,2)$, using the same witness $S_3$, the same normalized violation, and the same fixed generation input. The calculation is repeated independently for the six pairs $(x^\star,y^\star)\in\{0,1,2\}\times\{0,1\}$, with a separate observable $E_\star$ optimized for each pair, reflecting the fact that the public generation input is known before Eve performs her measurement.
 
For the exact physical models, classical side information is a special case of quantum side information, since the classical label can be encoded into orthogonal states of Eve's system. Hence,
\begin{align}
G^{\mathrm{cl}}(B|\Lambda,s_\star;s) & \leq G^{\mathrm Q}(B|E,s_\star;s), \\
H_{\min}^{\mathrm Q} &\leq H_{\min}^{\mathrm{cl}}.
\label{eq:qeve-physical-ordering}
\end{align}
The plotted quantities, however, are obtained from finite and different outer approximations, and a numerical separation between the classical and quantum curves does not, by itself, prove that the physical gap is strict. The informative comparison is therefore the one between the two quantum levels. Stability of the quantum certificate under the level increase supports the reading of a genuine side-information gap, whereas a substantial reduction of the separation indicates that part of it originates in the looseness of the lower relaxation.
 
\begin{figure}[t]
    \centering
    \includegraphics{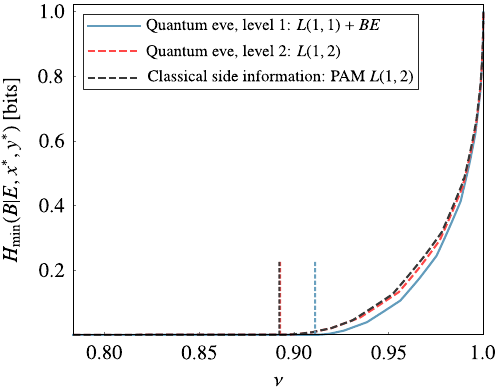}
    \caption{ \textbf{Classical and quantum side information at the generation input $(x^\star,y^\star)=(2,1)$.} Randomness certification from the witness $S_3$ in the prepare-and-measure architecture illustrated in Fig.~\ref{fig:quantum-broadcast}. The certified conditional min-entropy at the fixed generation input $(x^\star,y^\star)=(2,1)$ is shown as a function of the visibility $v$. The solid curves correspond to quantum side information at levels $L(1,1)+BE$ and $L(1,2)$, while the dashed curve corresponds to classical side information at the PAM level $L(1,2)$. The vertical lines mark the minimum visibility required for the corresponding certificate to become nontrivial.}
    \label{fig:s3-quantum-eve-all-inputs}
\end{figure}
 
The six generation inputs already play different roles in the reference qubit realization attaining $S_3^{\mathrm Q}$, since this single strategy fixes the bare output statistics of Bob's device before any adversary enters the picture, and because conditioning on Eve's system can only reduce Bob's uncertainty, this bare marginal upper bounds whatever entropy a certificate could ever report at that input. For $(x^\star,y^\star)\in\{(0,0),(0,1),(1,0),(1,1)\}$, the bare marginal satisfies $\max_b p_B(b|x^\star,y^\star)=\tfrac12(1+1/\sqrt2)$, the same bias as the optimal single-qubit QRAC, so that no certificate at those inputs can exceed $-\log_2[\tfrac12(1+1/\sqrt2)]\simeq0.2284$ bits, while the same realization renders the input $(2,0)$ deterministic, ruling out any randomness there, and leaves $(2,1)$ exactly uniform, the only one of the six inputs at which a full bit is available in principle.

In Fig.~\ref{fig:s3-quantum-eve-all-inputs} we focus on the generation input $(x^\star,y^\star)=(2,1)$, for which the reference realization allows the largest amount of randomness. At the $L(1,2)$ relaxation level, the classical- and quantum-side-information certificates are numerically close. Since the two curves are obtained from different finite outer approximations, this proximity should not be interpreted as a quantitative characterization of the gap between the exact physical models. The stability of the quantum certificate under the increase from $L(1,1)+BE$ to $L(1,2)$ is the more informative numerical comparison.

\section{Conclusions and outlook}
\label{sec:discussion}

We have further developed the prepare-and-broadcast scenario \cite{sarubi2026prepare} as a framework for studying quantum communication under dimension constraints. Beyond its original formulation, we have shown that it naturally captures the interplay between several manifestations of nonclassicality that cannot be addressed within standard prepare-and-measure experiments. In particular, we established quantitative trade-offs between prepare-and-measure witnesses, Bell nonlocality, and the communication advantages attainable by different receivers. The connection with asymmetric phase-covariant quantum cloning shows that these limitations are not merely properties of a particular communication task, but reflect fundamental restrictions on how the information encoded in a single quantum system can be distributed.

We also demonstrated that the prepare-and-broadcast scenario provides a natural setting for semi-device-independent randomness certification. The simultaneous access to marginal prepare-and-measure behaviors and multipartite correlations enables protocols that certify larger amounts of randomness than standard Bell-based approaches while remaining robust to noise. Moreover, the broadcast structure naturally accommodates stronger adversarial models, allowing an eavesdropper to retain quantum side information correlated with the measurement devices without abandoning the semi-device-independent paradigm.

Our results suggest several directions for future work. An immediate challenge is to extend the prepare-and-broadcast framework beyond communication-dimension constraints to alternative physical assumptions, such as bounded energy \cite{van2017semi}, imperfect dimension \cite{pauwels2022almost}, state fidelity \cite{tavakoli2021semi}, or information-theoretic constraints \cite{chaves2015device,tavakoli2020informationally}. It would also be interesting to investigate prepare-and-broadcast networks with more than two receivers, where new forms of monogamy and information-sharing trade-offs may emerge. Finally, the combination of communication constraints, broadcasting, and randomness certification opens the possibility of developing semi-device-independent cryptographic protocols and other information-processing tasks with no counterpart in the standard prepare-and-measure framework.

\acknowledgements
We acknowledge financial support from the EU Horizon Europe (QSNP, grant no. 101114043), the Danish National Research Foundation grant bigQ (DNRF 142), the Simons Foundation (Grant No. 1023171, R.C.), a guest professorship from the Otto M\o nsted Foundation, the Brazilian National Council for Scientific and Technological Development (CNPq, Grants No. 403181/2024-0 and 301687/2025-0), the National Institute of Science and Technology for Applied Quantum Computing through CNPq process No. 408884/2024-0, the Financiadora de Estudos e Projetos (Grant No. 1699/24 IIF-FINEP), SENAI CIMATEC/EMBRAPII through the Competence Center in Quantum Technologies -- Quantum Industrial Innovation (QuIIN), under the project ``Design e arquiteturas de Protocolos CV-QKD'' with resources from the MCTI PPI IoT/Manufatura 4.0, Grant No. 053/2023, a guest professorship from the Otto M\o{}nsted Foundation, and the Coordenação de Aperfeiçoamento de Pessoal de Nível Superior -- Brasil (CAPES) -- Finance Code 001. We also thank the High-Performance Computing Center (NPAD) at UFRN for providing computational resources. 

\textbf{ Code availability.} The code used to generate the results in this paper is available on Github: \href{https://github.com/Sarubi22/Randomness-Certification-and-Trade-offs-in-the-Prepare-and-Broadcast-Scenario}{https://github.com/Sarubi22/Randomness-Certification-and-Trade-offs-in-the-Prepare-and-Broadcast-Scenario}
\newpage

\bibliography{2-references}

@article{Jordan1875,
  author = {Jordan, Camille},
  title = {Essai sur la g{\^e}om{\^e}trie {\`a} $n$ dimensions},
  journal = {Bulletin de la Soci{\'e}t{\'e} Math{\'e}matique de France},
  year = {1875},
  volume = {3},
  pages = {103--174},
  doi = {10.24033/bsmf.90},
  url = {http://www.numdam.org/articles/10.24033/bsmf.90/}
}

@article{Masanes2006,
  author  = {Masanes, Llu{\'i}s},
  title   = {Asymptotic violation of {Bell} inequalities and distillability},
  journal = {Phys. Rev. Lett},
  volume  = {97},
  number  = {5},
  pages   = {050503},
  year    = {2006},
  doi     = {10.1103/PhysRevLett.97.050503},
  eprint  = {quant-ph/0512153},
  archivePrefix = {arXiv},
  primaryClass  = {quant-ph}
}

@article{scarani2005quantum,
  title = {Quantum Cloning},
  author = {Scarani, Valerio and Iblisdir, Sofyan and Gisin, Nicolas and Ac\'{\i}n, Antonio},
  journal = {Rev. Mod. Phys.},
  volume = {77},
  issue = {4},
  pages = {1225--1256},
  year = {2005},
  month = {Nov},
  publisher = {American Physical Society},
  doi = {10.1103/RevModPhys.77.1225},
  url = {https://link.aps.org/doi/10.1103/RevModPhys.77.1225}
}

@misc{sarubi2026prepare,
      title={The Prepare and Broadcast Scenario}, 
      author={Tailan S. Sarubi and Moisés Alves and Santiago Zamora and Vinícius F. Alves and A. de Oliveira Junior and Rafael Chaves},
      year={2026},
      eprint={2606.28632},
      archivePrefix={arXiv},
      primaryClass={quant-ph},
      url={https://arxiv.org/abs/2606.28632}, 
}

@article{fan2014quantum,
  title = {Quantum Cloning Machines and the Applications},
  author = {Fan, Heng and Wang, Yi-Nan and Jing, Li and Yue, Jie-Dong and Shi, Han-Duo and Zhang, Yong-Liang and Mu, Liang-Zhu},
  journal = {Physics Reports},
  volume = {544},
  issue = {3},
  pages = {241--322},
  year = {2014},
  month = {Sep},
  publisher = {Elsevier},
  issn = {0370-1573},
  doi = {10.1016/j.physrep.2014.06.004},
  url = {https://www.sciencedirect.com/science/article/pii/S0370157314002099}
}

@article{Pironio2010,
  title = {Random Numbers Certified by {Bell}'s Theorem},
  author = {Pironio, S. and Ac{\'i}n, A. and Massar, S. and Boyer de la Giroday, A. and Matsukevich, D. N. and Maunz, P. and Olmschenk, S. and Hayes, D. and Luo, L. and Manning, T. A. and Monroe, C.},
  journal = {Nature},
  volume = {464},
  issue = {7291},
  pages = {1021--1024},
  year = {2010},
  month = {Apr},
  publisher = {Nature Publishing Group},
  doi = {10.1038/nature09008},
  url = {https://doi.org/10.1038/nature09008}
}

@article{acin2012randomness,
  title={Randomness versus nonlocality and entanglement},
  author={Ac{\'\i}n, Antonio and Massar, Serge and Pironio, Stefano},
  journal={Phys. Rev. Lett},
  volume={108},
  number={10},
  pages={100402},
  year={2012},
  publisher={APS}
}

@article{Acin2016,
  title = {Optimal randomness certification from one entangled bit},
  author = {Ac\'{\i}n, Antonio and Pironio, Stefano and V\'ertesi, Tam\'as and Wittek, Peter},
  journal = {Phys. Rev. A},
  volume = {93},
  issue = {4},
  pages = {040102(R)},
  numpages = {5},
  year = {2016},
  month = {Apr},
  publisher = {American Physical Society},
  doi = {10.1103/PhysRevA.93.040102},
  url = {https://link.aps.org/doi/10.1103/PhysRevA.93.040102}
}

@article{Wooltorton2022,
  title = {Tight Analytic Bound on the Trade-Off between Device-Independent Randomness and Nonlocality},
  author = {Wooltorton, Lewis and Brown, Peter and Colbeck, Roger},
  journal = {Phys. Rev. Lett.},
  volume = {129},
  issue = {15},
  pages = {150403},
  numpages = {6},
  year = {2022},
  month = {Oct},
  publisher = {American Physical Society},
  doi = {10.1103/PhysRevLett.129.150403},
  url = {https://link.aps.org/doi/10.1103/PhysRevLett.129.150403}
}

@article{NietoSilleras2014,
doi = {10.1088/1367-2630/16/1/013035},
url = {https://doi.org/10.1088/1367-2630/16/1/013035},
year = {2014},
month = {jan},
publisher = {IOP Publishing},
volume = {16},
number = {1},
pages = {013035},
author = {Nieto-Silleras, O and Pironio, S and Silman, J},
title = {Using complete measurement statistics for optimal device-independent randomness evaluation},
journal = {New Journal of Physics},
}

@article{Bancal2014,
  title = {More Randomness from the Same Data},
  author = {Bancal, J.-D. and Sheridan, L. and Scarani, V.},
  journal = {New J. Phys.},
  volume = {16},
  issue = {3},
  pages = {033011},
  year = {2014},
  month = {Mar},
  publisher = {IOP Publishing},
  doi = {10.1088/1367-2630/16/3/033011},
  url = {https://doi.org/10.1088/1367-2630/16/3/033011}
}

@article{NPA2007,
  title = {Bounding the Set of Quantum Correlations},
  author = {Navascu\'es, Miguel and Pironio, Stefano and Ac\'{\i}n, Antonio},
  journal = {Phys. Rev. Lett.},
  volume = {98},
  issue = {1},
  pages = {010401},
  numpages = {4},
  year = {2007},
  month = {Jan},
  publisher = {American Physical Society},
  doi = {10.1103/PhysRevLett.98.010401},
  url = {https://link.aps.org/doi/10.1103/PhysRevLett.98.010401}
}

@article{NPA2008,
title = {A convergent hierarchy of semidefinite programs characterizing the set of quantum correlations},
author = {Navascués, Miguel and Pironio, Stefano and Acín, Antonio},
journal = {New Journal of Physics},
volume = {10},
pages = {073013},
year = {2008},
month = {jul},
publisher = {},
doi = {10.1088/1367-2630/10/7/073013},
url = {https://doi.org/10.1088/1367-2630/10/7/073013}
}

@article{NV2015,
  title = {Bounding the Set of Finite Dimensional Quantum Correlations},
  author = {Navascu\'es, Miguel and V\'ertesi, Tam\'as},
  journal = {Phys. Rev. Lett.},
  volume = {115},
  issue = {2},
  pages = {020501},
  numpages = {5},
  year = {2015},
  month = {Jul},
  publisher = {American Physical Society},
  doi = {10.1103/PhysRevLett.115.020501},
  url = {https://link.aps.org/doi/10.1103/PhysRevLett.115.020501}
}

@article{Stinespring1955,
  title = {Positive functions on  {$C^*$}-algebras},
  author = {Stinespring, W. F.},
  journal = {Proceedings of the American Mathematical Society},
  volume = {6},
  pages = {211--216},
  year = {1955},
  doi = {10.1090/s0002-9939-1955-0069403-4},
  url = {https://pubs.ams.org/journals/proc/1955-006-02/S0002-9939-1955-0069403-4}
}

@article{Naimark1940,
  author  = {Naimark, M. A.},
  title   = {Spectral functions of a symmetric operator},
  journal = {Izvestiya Akademii Nauk SSSR, Seriya Matematicheskaya},
  volume  = {4},
  number  = {3},
  pages   = {277--318},
  year    = {1940},
  url     = {http://mi.mathnet.ru/eng/im3893}
}

@article{brukner2004bell,
  title = {Bell's Inequalities and Quantum Communication Complexity},
  author = {Brukner, {\v C}aslav and {\.Z}ukowski, Marek and Pan, Jian-Wei and Zeilinger, Anton},
  journal = {Phys. Rev. Lett},
  volume = {92},
  number = {12},
  pages = {127901},
  year = {2004},
  month = {Mar},
  publisher = {American Physical Society},
  doi = {10.1103/PhysRevLett.92.127901},
  url = {https://link.aps.org/doi/10.1103/PhysRevLett.92.127901}
}

@article{Poderini2020,
  title = {Criteria for nonclassicality in the prepare-and-measure scenario},
  author = {Poderini, Davide and Brito, Samura\'{\i} and Nery, Ranieri and Sciarrino, Fabio and Chaves, Rafael},
  journal = {Phys. Rev. Res.},
  volume = {2},
  issue = {4},
  pages = {043106},
  numpages = {13},
  year = {2020},
  month = {Oct},
  publisher = {American Physical Society},
  doi = {10.1103/PhysRevResearch.2.043106},
  url = {https://link.aps.org/doi/10.1103/PhysRevResearch.2.043106}
}

@article{chaves2018causal,
  title = {Causal Modeling the Delayed-Choice Experiment},
  author = {Chaves, Rafael and Lemos, Gabriela Barreto and Pienaar, Jacques},
  journal = {Phys. Rev. Lett.},
  volume = {120},
  issue = {19},
  pages = {190401},
  numpages = {6},
  year = {2018},
  month = {May},
  publisher = {American Physical Society},
  doi = {10.1103/PhysRevLett.120.190401},
  url = {https://link.aps.org/doi/10.1103/PhysRevLett.120.190401}
}

@article{li2012semi,
  title = {Semi-device-independent randomness certification using $n\ensuremath{\rightarrow}1$ quantum random access codes},
  author = {Li, Hong-Wei and Paw\l{}owski, Marcin and Yin, Zhen-Qiang and Guo, Guang-Can and Han, Zheng-Fu},
  journal = {Phys. Rev. A},
  volume = {85},
  issue = {5},
  pages = {052308},
  numpages = {4},
  year = {2012},
  month = {May},
  publisher = {American Physical Society},
  doi = {10.1103/PhysRevA.85.052308},
  url = {https://link.aps.org/doi/10.1103/PhysRevA.85.052308}
}

@article{alves2026semi,
  title = {Semi-device-independent randomness certification on discretized continuous-variable platforms},
  author = {Alves, Mois\'es and Sena, Vitor L. and Zamora, Santiago and Sarubi, Tailan S. and de Oliveira Junior, A. and Tacla, Alexandre B. and Chaves, Rafael},
  journal = {Phys. Rev. A},
  volume = {113},
  issue = {4},
  pages = {042418},
  numpages = {17},
  year = {2026},
  month = {Apr},
  publisher = {American Physical Society},
  doi = {10.1103/z3s5-bhb7},
  url = {https://link.aps.org/doi/10.1103/z3s5-bhb7}
}

@article{pawlowski2011semi,
  title = {Semi-device-independent security of one-way quantum key distribution},
  author = {Paw\l{}owski, Marcin and Brunner, Nicolas},
  journal = {Phys. Rev. A},
  volume = {84},
  issue = {1},
  pages = {010302(R)},
  numpages = {4},
  year = {2011},
  month = {Jul},
  publisher = {American Physical Society},
  doi = {10.1103/PhysRevA.84.010302},
  url = {https://link.aps.org/doi/10.1103/PhysRevA.84.010302}
}

@article{vsupic2020self,
  title = {Self-Testing of Quantum Systems: A Review},
  author = {{\v{S}}upi{\'c}, Ivan and Bowles, Joseph},
  journal = {Quantum},
  volume = {4},
  pages = {337},
  year = {2020},
  month = {Sep},
  publisher = {Verein zur F{\"o}rderung des Open Access Publizierens in den Quantenwissenschaften},
  issn = {2521-327X},
  doi = {10.22331/q-2020-09-30-337},
  url = {https://doi.org/10.22331/q-2020-09-30-337}
}

@article{jia2025characterizing,
  title = {Characterizing the set of quantum correlations in prepare-and-measure quantum multichain-shaped networks},
  author = {Jia, Yanning and Guo, Fenzhuo and Wang, Yukun and Dong, Haifeng and Gao, Fei},
  journal = {Phys. Rev. A},
  volume = {111},
  issue = {2},
  pages = {022439},
  numpages = {11},
  year = {2025},
  month = {Feb},
  publisher = {American Physical Society},
  doi = {10.1103/PhysRevA.111.022439},
  url = {https://link.aps.org/doi/10.1103/PhysRevA.111.022439}
}

@article{clauser1969proposed,
  title = {Proposed Experiment to Test Local Hidden-Variable Theories},
  author = {Clauser, John F. and Horne, Michael A. and Shimony, Abner and Holt, Richard A.},
  journal = {Phys. Rev. Lett.},
  volume = {23},
  issue = {15},
  pages = {880--884},
  numpages = {0},
  year = {1969},
  month = {Oct},
  publisher = {American Physical Society},
  doi = {10.1103/PhysRevLett.23.880},
  url = {https://link.aps.org/doi/10.1103/PhysRevLett.23.880}
}

@article{wang2019characterising,
  title = {Characterising the correlations of prepare-and-measure quantum networks},
  author = {Wang, Yukun and Primaatmaja, Ignatius William and Lavie, Emilien and Varvitsiotis, Antonios and Lim, Charles Ci Wen},
  journal = {npj Quantum Information},
  volume = {5},
  number = {1},
  pages = {17},
  year = {2019},
  month = {Feb},
  publisher = {Nature Publishing Group},
  doi = {10.1038/s41534-019-0133-3},
  url = {https://doi.org/10.1038/s41534-019-0133-3}
}

@article{pauwels2022almost,
  title = {Almost Qudits in the Prepare-and-Measure Scenario},
  author = {Pauwels, Jef and Pironio, Stefano and Woodhead, Erik and Tavakoli, Armin},
  journal = {Phys. Rev. Lett.},
  volume = {129},
  issue = {25},
  pages = {250504},
  numpages = {7},
  year = {2022},
  month = {Dec},
  publisher = {American Physical Society},
  doi = {10.1103/PhysRevLett.129.250504},
  url = {https://link.aps.org/doi/10.1103/PhysRevLett.129.250504}
}

@article{tavakoli2021semi,
  title = {Semi-Device-Independent Framework Based on Restricted Distrust in Prepare-and-Measure Experiments},
  author = {Tavakoli, Armin},
  journal = {Phys. Rev. Lett.},
  volume = {126},
  issue = {21},
  pages = {210503},
  numpages = {6},
  year = {2021},
  month = {May},
  publisher = {American Physical Society},
  doi = {10.1103/PhysRevLett.126.210503},
  url = {https://link.aps.org/doi/10.1103/PhysRevLett.126.210503}
}

@article{chaves2015device,
  title = {Device-Independent Tests of Entropy},
  author = {Chaves, Rafael and Brask, Jonatan Bohr and Brunner, Nicolas},
  journal = {Phys. Rev. Lett.},
  volume = {115},
  issue = {11},
  pages = {110501},
  numpages = {5},
  year = {2015},
  month = {Sep},
  publisher = {American Physical Society},
  doi = {10.1103/PhysRevLett.115.110501},
  url = {https://link.aps.org/doi/10.1103/PhysRevLett.115.110501}
}

@article{tavakoli2020informationally,
  doi = {10.22331/q-2020-09-24-332},
  url = {https://doi.org/10.22331/q-2020-09-24-332},
  title = {Informationally restricted quantum correlations},
  author = {Tavakoli, Armin and Zambrini Cruzeiro, Emmanuel and Bohr Brask, Jonatan and Gisin, Nicolas and Brunner, Nicolas},
  journal = {{Quantum}},
  issn = {2521-327X},
  publisher = {{Verein zur F{\"{o}}rderung des Open Access Publizierens in den Quantenwissenschaften}},
  volume = {4},
  pages = {332},
  month = sep,
  year = {2020}
}

@article{van2017semi,
  doi = {10.22331/q-2017-11-18-33},
  url = {https://doi.org/10.22331/q-2017-11-18-33},
  title = {Semi-device-independent framework based on natural physical assumptions},
  author = {Van Himbeeck, Thomas and Woodhead, Erik and Cerf, Nicolas J. and Garc{\'{i}}a-Patr{\'{o}}n, Ra{\'{u}}l and Pironio, Stefano},
  journal = {{Quantum}},
  issn = {2521-327X},
  publisher = {{Verein zur F{\"{o}}rderung des Open Access Publizierens in den Quantenwissenschaften}},
  volume = {1},
  pages = {33},
  month = nov,
  year = {2017}
}

@article{ioannou2022receiver,
doi = {10.1088/1367-2630/ac71bc},
url = {https://doi.org/10.1088/1367-2630/ac71bc},
year = {2022},
month = {jun},
publisher = {IOP Publishing},
volume = {24},
number = {6},
pages = {063006},
author = {Ioannou, Marie and Sekatski, Pavel and Abbott, Alastair A and Rosset, Denis and Bancal, Jean-Daniel and Brunner, Nicolas},
title = {Receiver-device-independent quantum key distribution protocols},
journal = {New J. Phys}
}

@article{chaves2015information,
  title = {Information-theoretic implications of quantum causal structures},
  author = {Chaves, Rafael and Majenz, Christian and Gross, David},
  journal = {Nat. Commun},
  volume = {6},
  number = {1},
  pages = {5766},
  year = {2015},
  month = {Jan},
  publisher = {Nature Publishing Group},
  doi = {10.1038/ncomms6766},
  url = {https://doi.org/10.1038/ncomms6766}
}

@article{pawlowski2009information,
  title = {Information causality as a physical principle},
  author = {Paw{\l}owski, Marcin and Paterek, Tomasz and Kaszlikowski, Dagomir and Scarani, Valerio and Winter, Andreas and {\.Z}ukowski, Marek},
  journal = {Nature},
  volume = {461},
  number = {7267},
  pages = {1101--1104},
  year = {2009},
  month = {Oct},
  publisher = {Nature Publishing Group},
  doi = {10.1038/nature08400},
  url = {https://doi.org/10.1038/nature08400}
}

@phdthesis{woodhead2014imperfections,
  title = {Imperfections and self testing in prepare-and-measure quantum key distribution},
  author = {Woodhead, Erik},
  school = {Universit{\'e} libre de Bruxelles, Facult{\'e} des Sciences -- Physique},
  address = {Bruxelles},
  year = {2014},
  type = {PhD thesis},
  url = {http://hdl.handle.net/2013/ULB-DIPOT:oai:dipot.ulb.ac.be:2013/209185}
}

@article{lunghi2015self,
  title = {Self-Testing Quantum Random Number Generator},
  author = {Lunghi, Tommaso and Brask, Jonatan Bohr and Lim, Charles Ci Wen and Lavigne, Quentin and Bowles, Joseph and Martin, Anthony and Zbinden, Hugo and Brunner, Nicolas},
  journal = {Phys. Rev. Lett.},
  volume = {114},
  issue = {15},
  pages = {150501},
  numpages = {5},
  year = {2015},
  month = {Apr},
  publisher = {American Physical Society},
  doi = {10.1103/PhysRevLett.114.150501},
  url = {https://link.aps.org/doi/10.1103/PhysRevLett.114.150501}
}

@article{passaro2015optimal,
doi = {10.1088/1367-2630/17/11/113010},
url = {https://doi.org/10.1088/1367-2630/17/11/113010},
year = {2015},
month = {oct},
publisher = {IOP Publishing},
volume = {17},
number = {11},
pages = {113010},
author = {Passaro, Elsa and Cavalcanti, Daniel and Skrzypczyk, Paul and Acín, Antonio},
title = {Optimal randomness certification in the quantum steering and prepare-and-measure scenarios},
journal = {New J. Phys}
}

@misc{bohr2026quantum,
      title={Quantum correlations in prepare-and-measure scenarios and their semi-device-independent applications}, 
      author={Jonatan Bohr Brask and Nicolas Brunner and Jef Pauwels and Davide Rusca and Armin Tavakoli},
      year={2026},
      eprint={2603.23604},
      archivePrefix={arXiv},
      primaryClass={quant-ph},
      url={https://arxiv.org/abs/2603.23604}, 
}

@article{pauwels2022entanglement,
doi = {10.1088/1367-2630/ac724a},
url = {https://doi.org/10.1088/1367-2630/ac724a},
year = {2022},
month = {jun},
publisher = {IOP Publishing},
volume = {24},
number = {6},
pages = {063015},
author = {Pauwels, Jef and Tavakoli, Armin and Woodhead, Erik and Pironio, Stefano},
title = {Entanglement in prepare-and-measure scenarios: many questions, a few answers},
journal = {New J. Phys}
}

@article{vieira2023interplays,
doi = {10.1088/1367-2630/ad0526},
url = {https://doi.org/10.1088/1367-2630/ad0526},
year = {2023},
month = {nov},
publisher = {IOP Publishing},
volume = {25},
number = {11},
pages = {113004},
author = {Vieira, Carlos and de Gois, Carlos and Pollyceno, Lucas and Rabelo, Rafael},
title = {Interplays between classical and quantum entanglement-assisted communication scenarios},
journal = {New J. Phys}
}

@article{moreno2021semi,
  title = {Semi-device-independent certification of entanglement in superdense coding},
  author = {Moreno, George and Nery, Ranieri and de Gois, Carlos and Rabelo, Rafael and Chaves, Rafael},
  journal = {Phys. Rev. A},
  volume = {103},
  issue = {2},
  pages = {022426},
  numpages = {11},
  year = {2021},
  month = {Feb},
  publisher = {American Physical Society},
  doi = {10.1103/PhysRevA.103.022426},
  url = {https://link.aps.org/doi/10.1103/PhysRevA.103.022426}
}

@article{tavakoli2018semi,
  title = {Semi-device-independent characterization of multipartite entanglement of states and measurements},
  author = {Tavakoli, Armin and Abbott, Alastair A. and Renou, Marc-Olivier and Gisin, Nicolas and Brunner, Nicolas},
  journal = {Phys. Rev. A},
  volume = {98},
  issue = {5},
  pages = {052333},
  numpages = {9},
  year = {2018},
  month = {Nov},
  publisher = {American Physical Society},
  doi = {10.1103/PhysRevA.98.052333},
  url = {https://link.aps.org/doi/10.1103/PhysRevA.98.052333}
}

@article{gallego2010device,
  title = {Device-Independent Tests of Classical and Quantum Dimensions},
  author = {Gallego, Rodrigo and Brunner, Nicolas and Hadley, Christopher and Ac\'{\i}n, Antonio},
  journal = {Phys. Rev. Lett.},
  volume = {105},
  issue = {23},
  pages = {230501},
  numpages = {4},
  year = {2010},
  month = {Nov},
  publisher = {American Physical Society},
  doi = {10.1103/PhysRevLett.105.230501},
  url = {https://link.aps.org/doi/10.1103/PhysRevLett.105.230501}
}

@article{pawlowski2010entanglement,
  title = {Entanglement-assisted random access codes},
  author = {Paw\l{}owski, Marcin and \ifmmode \dot{Z}\else \.{Z}\fi{}ukowski, Marek},
  journal = {Phys. Rev. A},
  volume = {81},
  issue = {4},
  pages = {042326},
  numpages = {4},
  year = {2010},
  month = {Apr},
  publisher = {American Physical Society},
  doi = {10.1103/PhysRevA.81.042326},
  url = {https://link.aps.org/doi/10.1103/PhysRevA.81.042326}
}

@article{tavakoli2015quantum,
  title = {Quantum Random Access Codes Using Single $d$-Level Systems},
  author = {Tavakoli, Armin and Hameedi, Alley and Marques, Breno and Bourennane, Mohamed},
  journal = {Phys. Rev. Lett.},
  volume = {114},
  issue = {17},
  pages = {170502},
  numpages = {5},
  year = {2015},
  month = {Apr},
  publisher = {American Physical Society},
  doi = {10.1103/PhysRevLett.114.170502},
  url = {https://link.aps.org/doi/10.1103/PhysRevLett.114.170502}
}

@article{tavakoli2021correlations,
  title = {Correlations in Entanglement-Assisted Prepare-and-Measure Scenarios},
  author = {Tavakoli, Armin and Pauwels, Jef and Woodhead, Erik and Pironio, Stefano},
  journal = {PRX Quantum},
  volume = {2},
  issue = {4},
  pages = {040357},
  numpages = {17},
  year = {2021},
  month = {Dec},
  publisher = {American Physical Society},
  doi = {10.1103/PRXQuantum.2.040357},
  url = {https://link.aps.org/doi/10.1103/PRXQuantum.2.040357}
}

@article{zamora2025semi,
  title = {Semi-device-independent nonstabilizerness certification in the prepare-and-measure scenario},
  author = {Zamora, Santiago and Mac\^edo, Rafael A. and Sarubi, Tailan S. and Alves, Mois\'es and Poderini, Davide and Chaves, Rafael},
  journal = {Phys. Rev. A},
  volume = {112},
  issue = {4},
  pages = {042410},
  numpages = {14},
  year = {2025},
  month = {Oct},
  publisher = {American Physical Society},
  doi = {10.1103/wm7m-xnfq},
  url = {https://link.aps.org/doi/10.1103/wm7m-xnfq}
}

@article{Ambainis2002,
author = {Ambainis, Andris and Nayak, Ashwin and Ta-Shma, Amnon and Vazirani, Umesh},
title = {Dense quantum coding and quantum finite automata},
year = {2002},
issue_date = {July 2002},
publisher = {Association for Computing Machinery},
address = {New York, NY, USA},
volume = {49},
number = {4},
issn = {0004-5411},
url = {https://doi.org/10.1145/581771.581773},
doi = {10.1145/581771.581773},
journal = {J. ACM},
month = jul,
pages = {496–511},
numpages = {16}
}

@article{Nayak1999,
  title={Optimal lower bounds for quantum automata and random access codes},
  author={Ashwin Nayak},
  journal={40th Annual Symposium on Foundations of Computer Science (Cat. No.99CB37039)},
  year={1999},
  pages={369-376},
  url={https://api.semanticscholar.org/CorpusID:3486}
}

@article{Mohan2019,
doi = {10.1088/1367-2630/ab3773},
url = {https://doi.org/10.1088/1367-2630/ab3773},
year = {2019},
month = {aug},
publisher = {IOP Publishing},
volume = {21},
number = {8},
pages = {083034},
author = {Mohan, Karthik and Tavakoli, Armin and Brunner, Nicolas},
title = {Sequential random access codes and self-testing of quantum measurement instruments},
journal = {New J. Phys}
}

@article{Bruss2000,
  title = {Phase-covariant quantum cloning},
  author = {Bru\ss{}, Dagmar and Cinchetti, Mirko and Mauro D'Ariano, G. and Macchiavello, Chiara},
  journal = {Phys. Rev. A},
  volume = {62},
  issue = {1},
  pages = {012302},
  numpages = {7},
  year = {2000},
  month = {Jun},
  publisher = {American Physical Society},
  doi = {10.1103/PhysRevA.62.012302},
  url = {https://link.aps.org/doi/10.1103/PhysRevA.62.012302}
}

@article{DAriano2003,
  title = {Optimal phase-covariant cloning for qubits and qutrits},
  author = {D'Ariano, Giacomo Mauro and Macchiavello, Chiara},
  journal = {Phys. Rev. A},
  volume = {67},
  issue = {4},
  pages = {042306},
  numpages = {9},
  year = {2003},
  month = {Apr},
  publisher = {American Physical Society},
  doi = {10.1103/PhysRevA.67.042306},
  url = {https://link.aps.org/doi/10.1103/PhysRevA.67.042306}
}

@article{Rezakhani2005AsymmetricPhaseCovariant,
title = {Separability in asymmetric phase-covariant cloning},
journal = {Phys. Lett. A},
volume = {336},
number = {4},
pages = {278-289},
year = {2005},
issn = {0375-9601},
doi = {https://doi.org/10.1016/j.physleta.2004.12.015},
url = {https://www.sciencedirect.com/science/article/pii/S0375960104017219},
author = {A.T. Rezakhani and S. Siadatnejad and A.H. Ghaderi}
}

@book{BoydVandenberghe2004,
  author    = {Boyd, Stephen and Vandenberghe, Lieven},
  title     = {Convex Optimization},
  publisher = {Cambridge University Press},
  year      = {2004},
  doi       = {10.1017/CBO9780511804441}
}

@misc{sarubi2025pab,
  title        = {Randomness-Certification-and-Trade-offs-in-the-Prepare-and-Broadcast -- Numerical Codes},
  url = {https://github.com/Sarubi22/Randomness-Certification-and-Trade-offs-in-the-Prepare-and-Broadcast-Scenario},
  author ={ Tailan Sarubi},
  year = {2025},
}

@article{Chen2007AsymmetricPhaseCovariant,
  title = {Experimental realization of $\mathbf{1}\mathbf{\ensuremath{\rightarrow}}\mathbf{2}$ asymmetric phase-covariant quantum cloning},
  author = {Chen, Hongwei and Zhou, Xianyi and Suter, Dieter and Du, Jiangfeng},
  journal = {Phys. Rev. A},
  volume = {75},
  issue = {1},
  pages = {012317},
  numpages = {5},
  year = {2007},
  month = {Jan},
  publisher = {American Physical Society},
  doi = {10.1103/PhysRevA.75.012317},
  url = {https://link.aps.org/doi/10.1103/PhysRevA.75.012317}
}

@misc{Ambainis2008QRAC,
      title={Quantum Random Access Codes with Shared Randomness}, 
      author={Andris Ambainis and Debbie Leung and Laura Mancinska and Maris Ozols},
      year={2009},
      eprint={0810.2937},
      archivePrefix={arXiv},
      primaryClass={quant-ph},
      url={https://arxiv.org/abs/0810.2937}, 
}

\appendix

\section{Proof of Eqs.~\eqref{eq:criterion4prepA} -~\eqref{eq:criterion4prepB}}
\label{app:A}

In this appendix we show that a set of four preparations
$\overline\rho=\{\rho_0,\rho_1,\rho_2,\rho_3\}$ satisfies every relabeling of the inequalities
\begin{equation}
    S_4 = E_{00}+E_{01}+E_{10}-E_{11}-E_{20}+E_{21}-E_{30}-E_{31} \le 4,
    \label{eq:app:S4}
\end{equation}
and
\begin{equation}
    S_3 = E_{00}+E_{01}+E_{10}-E_{11}-E_{20} \le 3,
    \label{eq:app:S3}
\end{equation}
for arbitrary qubit measurements if and only if their Bloch vectors satisfy the following $15$ inequalities:
\begin{align}
    \|\vec{r}_0+\vec{r}_1-\vec{r}_2-\vec{r}_3\|_2 + \|\vec{r}_0-\vec{r}_1+\vec{r}_2-\vec{r}_3\|_2 &\leq 4,\nonumber\\
    \|\vec{r}_0+\vec{r}_1-\vec{r}_2-\vec{r}_3\|_2 +\|-\vec{r}_0+\vec{r}_1+\vec{r}_2-\vec{r}_3\|_2 &\leq 4,\label{eq:app:criterion4prepA}\\
    \|-\vec{r}_0+\vec{r}_1+\vec{r}_2-\vec{r}_3\|_2 +  \|\vec{r}_0-\vec{r}_1+\vec{r}_2-\vec{r}_3\|_2 &\leq 4,\nonumber
\end{align}
together with
\begin{align}
    \|\vec{r}_i+\vec{r}_j-\vec{r}_k\|_2+\|\vec{r}_i-\vec{r}_j\|_2 &\leq 3,\nonumber\\
    \|\vec{r}_i+\vec{r}_k-\vec{r}_j\|_2+\|\vec{r}_i-\vec{r}_k\|_2&\leq 3,\label{eq:app:criterion4prepB}\\
    \|\vec{r}_k+\vec{r}_j-\vec{r}_i\|_2+\|\vec{r}_k-\vec{r}_j\|_2 &\leq 3,\nonumber
\end{align}
for every $\{i,j,k\}\in\binom{\{0,1,2,3\}}{3}$. Since the last $12$ inequalities are exactly the three-preparation criterion of Ref.~\cite{Poderini2020} corresponding to all possible relabelings of $S_3$, applied to each subset of three preparations in $\overline\rho$, it remains to prove the three genuinely four-preparation inequalities. Consider the Bloch representation
$\rho_x=\frac12(\iden+\vec r_x\cdot\vec\sigma)$, with $x\in\{0,1,2,3\}$, and dichotomic qubit measurements
$B_{b|y}=\frac12(\iden+(-1)^b\vec q_y\cdot\vec\sigma)$, with $y\in\{0,1\}$. Then the correlators are given by $E_{xy}=\vec r_x\cdot\vec q_y$. Substituting this into Eq.~\eqref{eq:app:S4}, we obtain
\begin{equation}
    \vec q_0\cdot(\vec r_0+\vec r_1-\vec r_2-\vec r_3)  +  \vec q_1\cdot(\vec r_0-\vec r_1+\vec r_2-\vec r_3) \leq 4.
    \label{eq:app:S4_bloch}
\end{equation}
To find the maximal value for a fixed set of preparations, we optimize over Bob's measurement directions $\vec q_0$ and $\vec q_1$. Since these are Bloch vectors of dichotomic qubit observables, they satisfy $\|\vec q_y\|_2\leq 1$, with the maximum attained by projective measurements satisfying $\|\vec q_y\|_2=1$. The expression in Eq.~\eqref{eq:app:S4_bloch} is maximized by choosing $\vec q_0$ parallel to the first vector and $\vec q_1$ parallel to the second one. Therefore, Eq.~\eqref{eq:app:S4} is satisfied for all qubit measurements if and only if
\begin{equation}
     \|\vec r_0+\vec r_1-\vec r_2-\vec r_3\|_2 + \|\vec r_0-\vec r_1+\vec r_2-\vec r_3\|_2   \leq 4.
     \label{eq:app:first_4prep_condition}
\end{equation}
This is the first inequality in Eq.~\eqref{eq:app:criterion4prepA}. We now consider the remaining relabelings of Eq.~\eqref{eq:app:S4}. The relevant sign patterns have the form of a partition of the four preparations into two pairs, with one pair carrying positive signs and the other carrying negative signs. There are exactly three such pair partitions. Define
\begin{align}
    \vec w_1&:=\vec r_0+\vec r_1-\vec r_2-\vec r_3,\nonumber\\
    \vec w_2&:=\vec r_0-\vec r_1+\vec r_2-\vec r_3,\label{eq:app:w_vectors}\\
    \vec w_3&:=\vec r_0-\vec r_1-\vec r_2+\vec r_3.
    \nonumber
\end{align}
Thus, Eq.~\eqref{eq:app:first_4prep_condition} can be written as
\begin{equation}
    \|\vec w_1\|_2+\|\vec w_2\|_2\leq 4.
    \label{eq:app:w12_condition}
\end{equation}
Relabeling the preparations only permutes the three vectors
$\{\vec w_1,\vec w_2,\vec w_3\}$, up to irrelevant global signs inside the norms. Hence, the only inequivalent four-preparation conditions obtained from relabelings of Eq.~\eqref{eq:app:S4} are
\begin{align}
    \|\vec w_1\|_2+\|\vec w_2\|_2 &\leq 4,\nonumber\\
    \|\vec w_1\|_2+\|\vec w_3\|_2 &\leq 4,\label{eq:app:w_conditions}\\
    \|\vec w_2\|_2+\|\vec w_3\|_2 &\leq 4.\nonumber
\end{align}
After substituting the definitions in Eq.~\eqref{eq:app:w_vectors}, these three inequalities are exactly the three conditions in Eq.~\eqref{eq:criterion4prepA}. Together with the three-preparation conditions in Eq.~\eqref{eq:criterion4prepB}, they give the claimed necessary and sufficient criterion.

\section{Support-function methods}\label{app:support_function_lower_bounds}

The trade-off boundaries of Sec.~\ref{sec:broadcast_qrac_tradeoff} and the randomness curves of Sec.~\ref{sec:randomness} are produced by the same convex-analytic construction. In both cases the quantities of interest are mapped onto a convex subset of a two-dimensional plane, and linear support functions are maximized along a one-parameter family of directions. For the trade-offs, the two coordinates are linear scores evaluated on a single broadcast realization, and the support function is bounded from below by explicit strategies and from above by the moment relaxation. For randomness certification, one coordinate is the witness value and the other is the probability assigned to one guess of the adversary, and dualizing the observed witness constraint yields an upper envelope on the guessing probability, hence a lower bound on the certified min-entropy, valid at every witness value at once. Throughout, $\mathcal S_{q,t}(r)$ denotes the moment set of Sec.~\ref{sec:rand-framework}, at level $L(q,t)$ and total weight $r$, which is also the relaxation used for the trade-off upper bounds \cite{sarubi2026prepare}.

\subsection{Trade-off curves}
\label{app:tradeoff-support}

Let $G_1$ and $G_2$ be two linear scores of the broadcast behavior and let
\begin{equation}
    \mathcal R_Q
    =
    \operatorname{conv}
    \left\{
        \bigl(G_1[p],G_2[p]\bigr)\,:\,p\in\Qbr
    \right\}
    \label{eq:app-tradeoff-set}
\end{equation}
be the set of attainable pairs, convex because the devices may share classical randomness. Instead of fixing a target value of one coordinate and
maximizing the other, we optimize the support function
\begin{equation}
    h_Q(\mu)
    =
    \sup_{(u,v)\in\mathcal R_Q}\left(u+\mu v\right),
    \qquad
    \mu\geq0,
    \label{eq:app-tradeoff-support}
\end{equation}
so that varying $\mu\geq0$ exposes the upper Pareto boundary relevant to the simultaneous maximization of $G_1$ and $G_2$.

Lower bounds are obtained from an explicit qubit-output isometric ansatz, which generates physical strategies and is not the most general parametrization of the set constrained by the relaxation. Within this ansatz Alice prepares $\rho_x\in\mathcal D(\mathbb C^2)$, the broadcast map is an isometry $U:\mathbb C^2\to\mathbb C^2\otimes\mathbb C^2$ with $U^\dagger U=\id_2$, and Bob and Charlie perform the local projective measurements $B_y=\vec b_y\cdot\vec\sigma$ and $C_z=\vec c_z\cdot\vec\sigma$ on $\rho^{BC}_x=U\rho_xU^\dagger$, giving the correlators $\avg{B_y}{x}=\tr[\rho^{BC}_x(B_y\otimes\id_2)]$, $\avg{C_z}{x}=\tr[\rho^{BC}_x(\id_2\otimes C_z)]$, and $\corr{y}{z}{x}=\tr[\rho^{BC}_x(B_y\otimes C_z)]$. Each score is then written as
\begin{align}
G_i
={}&
\kappa_i
+
\sum_{x,y}a^i_{xy}\avg{B_y}{x}
+
\sum_{x,z}b^i_{xz}\avg{C_z}{x}
\nonumber\\
&+
\sum_{x,y,z}c^i_{xyz}\corr{y}{z}{x},
\qquad i=1,2,
\label{eq:app-tradeoff-coordinates}
\end{align}
where the constants $\kappa_i$ accommodate affine scores such as success
probabilities. The weighted score $G_\mu=G_1+\mu G_2$ has coefficients
$a^\mu_{xy}=a^1_{xy}+\mu a^2_{xy}$, and analogously for $b^\mu_{xz}$,
$c^\mu_{xyz}$, and $\kappa_\mu$.

At fixed $U$, $B_y$, and $C_z$ the dependence on each preparation is linear,
$G_\mu=\kappa_\mu+\sum_x\tr(\rho_xK^\mu_x)$, with the effective operators
\begin{align}
K^\mu_x
={}&
\sum_y a^\mu_{xy}\,U^\dagger(B_y\otimes\id_2)U
+
\sum_z b^\mu_{xz}\,U^\dagger(\id_2\otimes C_z)U
\nonumber\\
&+
\sum_{y,z} c^\mu_{xyz}\,U^\dagger(B_y\otimes C_z)U .
\label{eq:app-tradeoff-effective}
\end{align}
The optimization over the preparations is therefore analytic and returns $\max_{\rho_x}\tr(\rho_xK^\mu_x)=\lambda_{\max}(K^\mu_x)$, an optimal state being any state supported in the maximal eigenspace of $K^\mu_x$. The remaining search over the isometry, parametrized as $U=A(A^\dagger A)^{-1/2}$ with $A$ a complex $4\times2$ matrix, and over the unit Bloch vectors of the two receivers, is nonconvex and is performed numerically. For each multiplier we denote by $p_\mu$ the best explicit strategy found, whose support value
\begin{equation}
    \widehat h_{\mathrm{LB}}(\mu)
    :=
    G_1[p_\mu]+\mu G_2[p_\mu]
    \label{eq:app-tradeoff-lower}
\end{equation}
is attained by a feasible physical realization and therefore satisfies $\widehat h_{\mathrm{LB}}(\mu)\leq h_Q(\mu)$, with no claim that the
nonconvex search has reached its global maximum.

Upper bounds follow from the same support function evaluated over the moment relaxation. Writing $G_i[L]$ for the lifted scores,
\begin{equation}
    h^{(q,t)}_{\mathrm{UB}}(\mu)
    =
    \max_{L\in\mathcal S_{q,t}(1)}
    \left(
        G_1[L]+\mu G_2[L]
    \right),
    \label{eq:app-tradeoff-upper}
\end{equation}
which is again a single semidefinite program. Every physical strategy generates a feasible moment functional, so
\begin{equation}
    \widehat h_{\mathrm{LB}}(\mu)
    \leq
    h_Q(\mu)
    \leq
    h^{(q,t)}_{\mathrm{UB}}(\mu)
    \leq
    \widehat h^{(q,t)}_{\mathrm{UB}}(\mu),
    \label{eq:app-tradeoff-chain}
\end{equation}
the last inequality holding within the numerical accuracy stated above. Each multiplier thus supplies the outer supporting line
$u+\mu v\leq\widehat h^{(q,t)}_{\mathrm{UB}}(\mu)$ valid for every physical pair $(u,v)$. Agreement between the two ends of
Eq.~\eqref{eq:app-tradeoff-chain} within the stated tolerance identifies the support value in that direction at the relaxation level used, and, when the exposed optimizer is unique, the corresponding boundary point within the same tolerance. The markers displayed in Figs.~\ref{fig:combined_tradeoffs} and~\ref{fig:broadcast_qrac_tradeoff} are explicit-strategy points on the one hand and boundary points of the relaxed outer approximation, read off the optimizers of Eq.~\eqref{eq:app-tradeoff-upper}, on the other. In Sec.~\ref{sec:broadcast_qrac_tradeoff} the pair $(G_1,G_2)$ is either the two marginal QRAC scores $(P_B,P_C)$, or the marginal prepare-and-measure witnesses of the two receivers, or one marginal witness together with the conditioned Bell expression $\mathrm{CHSH}^{[x]}_{BC}$.

\subsection{Randomness envelopes}
\label{app:randomness-support}

We now apply the construction to the guessing-probability curves of Sec.~\ref{sec:randomness}, stating it for the joint randomness of the broadcast pair, the marginal case following by restriction. Solving the coupled branch program of Eq.~\eqref{eq:guessing-sdp-framework} pointwise certifies the guessing probability only at the sampled witness values, while evaluating the support function of each branch produces a single envelope valid for all of them at once. For each guess $g=(b_g,c_g)\in\mathcal G_{BC}$ and each multiplier $\tau\geq0$ we define
\begin{equation}
    F_g(\tau)
    =
    \max_{L\in\mathcal S_{q,t}(1)}
    \left[
        p_L(b_g,c_g|s)+\tau W[L]
    \right],
    \label{eq:support-F}
\end{equation}
a normalized single-branch semidefinite program of the same size as a standard witness optimization, which is the support function of the convex set $\bigl\{(W[L],p_L(b_g,c_g|s)):L\in\mathcal S_{q,t}(1)\bigr\}$ evaluated along the direction $(\tau,1)$. The construction rests on supporting hyperplanes and weak Lagrangian duality~\cite{BoydVandenberghe2004}, and dual bounds of this type, in which a linear functional yields an affine upper bound on the predictive power of the adversary, are standard in device-independent randomness certification~\cite{NietoSilleras2014,Bancal2014}.

Let $\{L_g\}_{g\in\mathcal G_{BC}}$ be feasible for Eq.~\eqref{eq:guessing-sdp-framework}, so that $L_g\in\mathcal S_{q,t}(r_g)$, $\sum_gr_g=1$, and $\sum_gW[L_g]\geq w$. The relaxed cones are homogeneous in the branch weight, in the sense that $L\in\mathcal S_{q,t}(r)$ with $r>0$ if and only if $L=rL'$ with $L'\in\mathcal S_{q,t}(1)$, the branch of vanishing weight being the zero functional. For every $\tau\geq0$,
\begin{align}
\sum_g p_{L_g}(b_g,c_g|s)
&\leq
\sum_g
\left[
    p_{L_g}(b_g,c_g|s)+\tau W[L_g]
\right]
-\tau w
\nonumber\\
&=
\sum_g r_g
\left[
    p_{L'_g}(b_g,c_g|s)+\tau W[L'_g]
\right]
-\tau w
\nonumber\\
&\leq
\max_{g\in\mathcal G_{BC}}F_g(\tau)-\tau w,
\label{eq:support-weak-duality}
\end{align}
the first line using the observed witness constraint and the last one using the normalization $\sum_gr_g=1$. Minimizing over the multiplier gives the
support-function envelope
\begin{equation}
    G_{q,t}(BC|\Lambda,s;w)
    \leq
    \inf_{\tau\geq0}
    \left[
        \max_{g\in\mathcal G_{BC}}F_g(\tau)-\tau w
    \right],
    \label{eq:support-envelope}
\end{equation}
which, combined with the inclusion of the physical set in its moment relaxation, yields the certified chain
\begin{align}
    G(BC|\Lambda,s;w)
    &\leq
    G_{q,t}(BC|\Lambda,s;w)
    \nonumber\\
    &\leq
    \inf_{\tau\geq0}
    \left[
        \max_{g\in\mathcal G_{BC}}F_g(\tau)-\tau w
    \right].
    \label{eq:certified-support-chain}
\end{align}
The first inequality holds because the relaxation is an outer approximation of the dimension-bounded PAB set, and the second follows from Eq.~\eqref{eq:support-weak-duality}, so strong duality is never invoked. At fixed $\tau$ the right-hand side is affine and nonincreasing in $w$, hence the envelope is a concave nonincreasing function of the observed witness value, and it is an upper bound on the guessing probability rather than an estimate of it from below. The certified min-entropy is therefore
\begin{equation}
    H_{\min}(BC|\Lambda,s)
    \geq
    -\log_2
    \left[
        \inf_{\tau\geq0}
        \left(
            \max_{g}F_g(\tau)-\tau w
        \right)
    \right],
    \label{eq:support-hmin}
\end{equation}
whenever the argument of the logarithm lies in $(0,1]$.

For the broadcast witness $W$ of Eq.~\eqref{eq:CC2}, an exact output symmetry halves the number of joint-guess support problems. The global relabeling $b\mapsto b\oplus1$ and $c\mapsto c\oplus1$, applied to every setting of both receivers, acts on the dichotomic observables as $B_y\mapsto-B_y$ and $C_z\mapsto-C_z$. Every correlator entering $W$ is a joint one and is therefore invariant under the simultaneous sign flip, while the transformation preserves positivity, normalization, the preparation constraints, and the measurement relations $B_y^2=C_z^2=\id$ and $[B_y,C_z]=0$, so it is a bijection of the feasible moment functionals at any relaxation level and exchanges the guesses $00\leftrightarrow11$ and $01\leftrightarrow10$ at the fixed generation input. Consequently
\begin{equation}
    F_{00}(\tau)=F_{11}(\tau),
    \qquad
    F_{01}(\tau)=F_{10}(\tau),
    \label{eq:support-guess-symmetry}
\end{equation}
and the maximum appearing in
Eqs.~\eqref{eq:support-envelope}--\eqref{eq:support-hmin} reduces to $\max\{F_{00}(\tau),F_{01}(\tau)\}$. The even-parity class $\{00,11\}$ and the odd-parity class $\{01,10\}$ remain inequivalent in general, since flipping the output of a single receiver reverses the sign of the joint correlators in $W$. The reduction is specific to this witness, and it is not needed in the marginal prepare-and-measure case, where the guess set has two elements from the start.

The certificates of Sec.~\ref{sec:randomness} are reported as functions of the visibility rather than of the raw witness value. For the noisy families used throughout, $p_v=v\,p_{\mathrm{opt}}+(1-v)\,p_{\mathrm{unif}}$ with $p_{\mathrm{unif}}$ the uniform behavior on the relevant outcome set, every term of the witnesses considered here vanishes on $p_{\mathrm{unif}}$, since $W$ contains only two-body correlators while $S_3$ and $S_4$ contain only single-party ones. Hence
\begin{equation}
    W[p_v]=v\,W_Q,
    \qquad
    v=\frac{w}{W_Q},
    \label{eq:app-visibility}
\end{equation}
so the visibility is the mixing parameter of the family for the fixed-distribution certificates and the observed witness value in units of
its maximal quantum value for the witness-only certificates, the two readings agreeing along the noisy family. The same identification applies to $S_3$, $S_4$, and $\mathrm{CHSH}^{[1]}$ with their respective quantum maxima, which turns the visibility into a common axis for witnesses of different classical bounds. A witness-only certificate is necessarily trivial whenever
\begin{equation}
    v\leq v_{\mathrm C}:=\frac{W^{\mathrm C}}{W_Q},
    \label{eq:app-classical-threshold}
\end{equation}
with $W^{\mathrm C}$ the classical bound of the witness in use, because a deterministic classical strategy attaining that bound satisfies the observed witness constraint while allowing the adversary to guess the outcome with certainty. The conclusion does not extend automatically to a fixed-distribution certificate, whose threshold depends on the complete observed behavior rather than on a single witness value.

In the numerical implementation the support functions are evaluated on a finite grid $\{\tau_k\}$, refined adaptively where the certificate varies fastest, and each value $\widehat F_g(\tau_k)$ is the conservative estimate described at the beginning of this appendix. Evaluating
Eq.~\eqref{eq:support-envelope} at $w=v\,W_Q$ and clipping to the range allowed on physical grounds, the envelope actually reported is
\begin{equation}
\begin{aligned}
    \widehat G^{\mathrm{gr}}_{q,t}(v)
    &=
    \min_k
    \left[
        \max_g\widehat F_g(\tau_k)-\tau_k v\,W_Q
    \right],
    \\
    \widehat G_{q,t}(v)
    &=
    \min\left\{
        1,\,
        \max\left\{
            \tfrac{1}{|\mathcal G|},\,
            \widehat G^{\mathrm{gr}}_{q,t}(v)
        \right\}
    \right\},
\end{aligned}
    \label{eq:finite-grid-envelope}
\end{equation}
with $|\mathcal G|=4$ for the joint randomness of the pair and $|\mathcal G|=2$ for a marginal. The lower clipping is legitimate because the adversary always guesses a most likely outcome, so $G\geq1/|\mathcal G|$ whenever the feasible set is nonempty, and it prevents residual numerical error from returning a min-entropy above the algebraic maximum. Restricting the minimization to finitely many multipliers can only increase the right-hand side relative to the continuous infimum, so a coarse grid weakens the certificate without ever strengthening it artificially. Once the witness constraint is dualized the branches decouple, and for the broadcast witness the symmetry above leaves two inequivalent guess classes, so a complete curve $v\mapsto\widehat G_{q,t}(v)$ at a fixed generation input follows from $2\,|\{\tau_k\}|$ single-branch programs rather than from one coupled four-branch program per sampled visibility. The saving matters most near the maximal violation, where the certified randomness varies rapidly and a pointwise treatment would demand fine sampling precisely in the numerically hardest regime.

For the quantum adversary of Sec.~\ref{sec:qeve} the same construction applies to a single normalized moment matrix rather than to a family of guess-labeled branches, in keeping with the coherent nature of the side information, and it is convenient to center the support function at the maximal violation, which in the parametrization of Eq.~\eqref{eq:app-visibility} amounts to centering at $v=1$. With $G_{s_\star}[L^{(\ell)}]$ the lifted objective of Eq.~\eqref{eq:qeve-guessing-correlator} and $v[L^{(\ell)}]=S_3[L^{(\ell)}]/S_3^{\mathrm Q}$ the lifted visibility, for each level $\ell$, each fixed generation input $s_\star$, and each $\tau\geq0$ we evaluate
\begin{equation}
q_{\ell,s_\star}(\tau)
=
\max_{L^{(\ell)}}
\left[
    G_{s_\star}[L^{(\ell)}]
    +
    \tau\bigl(v[L^{(\ell)}]-1\bigr)
\right],
\label{eq:qeve-centered-support}
\end{equation}
which shares its optimizer with the uncentered support function while avoiding the subtraction of two large and nearly equal quantities close to the endpoint. Weak duality then gives the uniform bound
\begin{equation}
G^{\mathrm Q}_{\ell}(B|E,s_\star;v)
\leq
\inf_{\tau\geq0}
\left[
    q_{\ell,s_\star}(\tau)
    +
    \tau(1-v)
\right].
\label{eq:qeve-support-envelope}
\end{equation}
The infimum is again evaluated over an adaptively refined finite grid, whichcan only weaken the certificate, and the reported curve is clipped from belowat $1/2$ for the same reason as in Eq.~\eqref{eq:finite-grid-envelope}. Direct pointwise semidefinite programs close to the maximal violation serve as independent conservative anchors, the exact endpoint is excluded from the numerical grid, and no analytic endpoint value is inserted.

\section{Two-bit certification at the maximal violation of \texorpdfstring{$W_{\mathcal{PAB}}$}{W	extsubscript{PAB}}}
\label{app:endpoint}
 
This appendix proves Theorem~\ref{rand:thm:main} and Corollary~\ref{rand:cor:orbit}. Throughout, $s^*=(1,0,0)$ is the generation input, and all statements refer to the adversarial model of Definition~\ref{rand:def:eve}. We write the witness \eqref{eq:CC2} as a sum of three preparation blocks,
\begin{equation}\label{rand:eq:blocks}
W_{\mathcal{PAB}} = W^{[0]}+W^{[1]}+W^{[2]},\qquad W^{[x]}=\avg{M_x}{x},
\end{equation}
with block operators and block bounds
\begin{equation}\label{rand:eq:blockops}
\begin{aligned}
M_0&=2\,(B_0C_0+B_1C_1), &\lambda_0&=4,\\
M_1&=2\,(B_0C_1-B_1C_0), &\lambda_1&=4,\\
M_2&=\B, &\lambda_2&=2\sqrt2,\\
\B&:=-B_0C_0-B_0C_1+B_1C_0-B_1C_1, &&
\end{aligned}
\end{equation}
where $\avg{O}{x}$ denotes the expectation of $O$ on the dilated preparation $x$ (defined below). The block deficits are $\delta_x=\lambda_x-W^{[x]}\ge0$, their nonnegativity being part of Lemma~\ref{rand:lem:bound}, and the total deficit $\sum_x\delta_x=\WQ-W_{\mathcal{PAB}}[p]$ vanishes exactly when $p$ saturates the maximal violation.
 
\subsection{Exact dilation to the projective isometric form}
 
\begin{lemma}[exact dilation]\label{rand:lem:dilation}
Every $p\in\Qbr$ admits a representation
\begin{equation}\label{rand:eq:dilrep}
p(b,c|x,y,z)=\tr\!\big[V\rho_x V^\dagger\,\big(\Pi^B_{b|y}\,\Pi^C_{c|z}\otimes \id_R\big)\big],
\end{equation}
with qubit states $\rho_x$, an isometry $V:\mathbb{C}^2\to \mathcal{H}_B\otimes\mathcal{H}_C\otimes\mathcal{H}_R$ between finite-dimensional spaces, and projective measurements whose reflections $B_y=\Pi^B_{0|y}-\Pi^B_{1|y}$ and $C_z=\Pi^C_{0|z}-\Pi^C_{1|z}$ act on the $B$ and $C$ tensor factors and obey
\begin{equation}\label{rand:eq:algebra}
B_y^2=C_z^2=\id,\qquad [B_y,C_z]=0 .
\end{equation}
The witness value and the generation distribution are unchanged.
\end{lemma}
 
\begin{proof}
By the Stinespring theorem \cite{Stinespring1955}, the broadcast channel dilates as $\Lambda(\rho)=\tr_E[V_0\rho V_0^\dagger]$ with an isometry $V_0:\mathbb{C}^2\to\mathcal{H}_B'\otimes\mathcal{H}_C'\otimes\mathcal{H}_E$ and $\dim\mathcal{H}_E$ finite, since the Kraus rank of a channel between finite-dimensional spaces is finite. For Bob's binary POVMs $\{E_{0|y},\id-E_{0|y}\}$ we use a single qubit ancilla, common to both settings, following Naimark \cite{Naimark1940}: on $\mathcal{H}_B'\otimes\mathbb{C}^2$, in $2\times2$ block form,
\begin{equation}\label{rand:eq:naimark}
\hat\Pi_{0|y}=\begin{pmatrix} E_{0|y} & D_y\\ D_y & \id-E_{0|y}\end{pmatrix},\quad
D_y=\sqrt{E_{0|y}(\id-E_{0|y})},
\end{equation}
and $\hat\Pi_{1|y}=\id-\hat\Pi_{0|y}$. Using $[E_{0|y},D_y]=0$ and $D_y^2=E_{0|y}(\id-E_{0|y})$, the diagonal blocks of $\hat\Pi_{0|y}^2$ are $E_{0|y}^2+D_y^2=E_{0|y}$ and its complement, and the off-diagonal blocks are $E_{0|y}D_y+D_y(\id-E_{0|y})=D_y$, so $\hat\Pi_{0|y}$ is a projector; moreover $J_B^\dagger\hat\Pi_{0|y}J_B=E_{0|y}$ for the embedding $J_B=\id\otimes|0\rangle$, which is the same for both settings and all inputs. The analogous construction for Charlie uses $J_C=\id\otimes|0\rangle$ on $\mathcal{H}_C'\otimes\mathbb{C}^2$. Setting $\mathcal{H}_B=\mathcal{H}_B'\otimes\mathbb{C}^2$, $\mathcal{H}_C=\mathcal{H}_C'\otimes\mathbb{C}^2$, $\mathcal{H}_R=\mathcal{H}_E$, and $V=(J_B\otimes J_C\otimes\id_E)V_0$, every joint probability is reproduced exactly. The dilated reflections act on different tensor factors, so \eqref{rand:eq:algebra} holds.
\end{proof}
 
From now on we work in the representation \eqref{rand:eq:dilrep}. We write $|\phi_x\rangle=V|\psi_x\rangle$ when $\rho_x=|\psi_x\rangle\langle\psi_x|$ is pure, $\avg{O}{x}=\tr[V\rho_xV^\dagger O]$ in general, and $\mathcal{S}=V(\mathbb{C}^2)$ for the two-dimensional code subspace, which contains every $|\phi_x\rangle$. Since the measurements are projective,
\begin{equation}\label{rand:eq:probform}
\begin{aligned}
p(b,c|x,y,z)=\tfrac14\big[1&+(-1)^b\avg{B_y}{x}+(-1)^c\avg{C_z}{x}\\
&+(-1)^{b+c}\avg{B_yC_z}{x}\big].
\end{aligned}
\end{equation}
 
\subsection{Block bounds, attainability, and saturation}
 
\begin{lemma}[bounds, attainability, saturation]\label{rand:lem:bound}
In the representation \eqref{rand:eq:dilrep} with receivers of arbitrary finite dimension, the three blocks obey $W^{[0]}\le4$, $W^{[1]}\le4$, and $W^{[2]}\le2\sqrt2$, so that $W_{\mathcal{PAB}}\le\WQ$, and the value $\WQ$ is attained by an explicit qubit strategy for which $p(b,c|s^*)=\tfrac14$. Moreover, if the preparations are pure and $W_{\mathcal{PAB}}=\WQ$, then
\begin{align}
B_0C_0|\phi_0\rangle&=|\phi_0\rangle, & B_1C_1|\phi_0\rangle&=|\phi_0\rangle, \label{rand:eq:S1}\\
B_0C_1|\phi_1\rangle&=|\phi_1\rangle, & B_1C_0|\phi_1\rangle&=-|\phi_1\rangle, \label{rand:eq:S2}\\
\{B_0,B_1\}|\phi_2\rangle&=0, & \{C_0,C_1\}|\phi_2\rangle&=0. \label{rand:eq:S3}
\end{align}
\end{lemma}
 
\begin{proof}
We first establish the three block bounds. Each $B_yC_z$ is a product of two commuting Hermitian involutions, hence a Hermitian involution of unit norm, and the triangle inequality gives $\|B_0C_0+B_1C_1\|\le2$ and $\|B_0C_1-B_1C_0\|\le2$, so $W^{[0]},W^{[1]}\le4$. For the third block put $A_0=-B_0$ and $A_1=B_1$, so that $\B=A_0(C_0+C_1)+A_1(C_0-C_1)$. Using $A_i^2=C_z^2=\id$, $[A_i,C_z]=0$, $(C_0+C_1)^2+(C_0-C_1)^2=4\,\id$, and $(C_0\pm C_1)(C_0\mp C_1)=\mp[C_0,C_1]$,
\begin{equation}\label{rand:eq:Bsquared}
\B^2 = 4\,\id-[A_0,A_1]\,[C_0,C_1].
\end{equation}
Since both commutators have norm at most $2$, $\|\B\|\le2\sqrt2$ and $W^{[2]}\le2\sqrt2$. The three blocks involve three different preparations, so each expectation is bounded separately and $W_{\mathcal{PAB}}\le8+2\sqrt2$.
 
The bound is attained on $\mathbb{C}^2\otimes\mathbb{C}^2$ by $B_0=C_0=Z$, $B_1=C_1=X$ (acting on the respective factors) and
\begin{equation}\label{rand:eq:strategy}
\begin{aligned}
|\phi_0\rangle&=\tfrac{1}{\sqrt2}\big(|00\rangle+|11\rangle\big),\\
|\phi_1\rangle&=\tfrac12\big(|00\rangle+|01\rangle-|10\rangle+|11\rangle\big),\\
|\phi_2\rangle&=\big(|\phi_0\rangle-|\phi_1\rangle\big)\big/\sqrt{2-\sqrt2}.
\end{aligned}
\end{equation}
Relations \eqref{rand:eq:S1}--\eqref{rand:eq:S2} hold by inspection, $\langle\phi_0|\phi_1\rangle=1/\sqrt2$ makes $|\phi_2\rangle$ a unit vector, and a direct evaluation in the computational basis gives the exact eigenvector relation
\begin{equation}\label{rand:eq:eigvec}
\begin{split}
\B\,\big(|\phi_0\rangle-|\phi_1\rangle\big)
&=\big(2-\sqrt2,\,-\sqrt2,\,\sqrt2,\,2-\sqrt2\big)^{\!\top}\\
&=2\sqrt2\,\big(|\phi_0\rangle-|\phi_1\rangle\big),
\end{split}
\end{equation}
so $W^{[2]}=2\sqrt2$ and $W_{\mathcal{PAB}}=\WQ$. The three vectors span the two-dimensional $\mathcal{S}=\operatorname{span}\{|\phi_0\rangle,|\phi_1\rangle\}$, so any isometry $V$ with range $\mathcal{S}$ and $|\psi_x\rangle=V^\dagger|\phi_x\rangle$ realizes a legitimate $d=2$ PAB strategy at the endpoint. At $s^*$ all three moments $\avg{B_0}{1}$, $\avg{C_0}{1}$, $\avg{B_0C_0}{1}$ vanish, so $p(b,c|s^*)=\tfrac14$ by \eqref{rand:eq:probform}, and the feasible set of \eqref{eq:guessing-general-framework} at $w=\WQ$ is nonempty.
 
Finally, we derive the saturation relations. Because the blocks are bounded independently, $W_{\mathcal{PAB}}=\WQ$ forces $W^{[0]}=W^{[1]}=4$ and $W^{[2]}=2\sqrt2$. Then $\avg{B_0C_0}{0}+\avg{B_1C_1}{0}=2$ with each term in $[-1,1]$ forces both to equal $1$; for a Hermitian involution $M$ and a unit vector with $\langle M\rangle=1$, the spectral split $M=M_+-M_-$ gives $\|M_-\phi\|=0$ and $M|\phi\rangle=|\phi\rangle$, which proves \eqref{rand:eq:S1}, and the signs in \eqref{eq:CC2} give \eqref{rand:eq:S2}. For \eqref{rand:eq:S3}, set $u=\tfrac{1}{\sqrt2}(C_0+C_1)$ and $v=\tfrac{1}{\sqrt2}(C_0-C_1)$ and verify the operator identity
\begin{equation}\label{rand:eq:SOS}
2\sqrt2\,\id-\B=\tfrac{1}{\sqrt2}\,(A_0-u)^2+\tfrac{1}{\sqrt2}\,(A_1-v)^2,
\end{equation}
which follows from $(A_i-w)^2=\id+w^2-2A_iw$ for $w\in\{u,v\}$, $u^2+v^2=2\,\id$, and $A_0u+A_1v=\tfrac{1}{\sqrt2}\B$. At $W^{[2]}=2\sqrt2$ the left side has zero expectation on $|\phi_2\rangle$ and both summands are positive semidefinite, so $A_0|\phi_2\rangle=u|\phi_2\rangle$ and $A_1|\phi_2\rangle=v|\phi_2\rangle$. Applying $A_0$ once more and using $u^2=\id+\tfrac12\{C_0,C_1\}$ gives $\{C_0,C_1\}|\phi_2\rangle=0$; and $\{A_0,A_1\}|\phi_2\rangle=\{u,v\}|\phi_2\rangle=\tfrac12\big(-[C_0,C_1]+[C_0,C_1]\big)|\phi_2\rangle=0$, so $\{B_0,B_1\}|\phi_2\rangle=0$.
\end{proof}
 
\subsection{Generation moments and the guessing bound}
 
The distribution at $s^*$ involves the three moments $\avg{B_0}{1}$, $\avg{C_0}{1}$, and $\avg{B_0C_0}{1}$, none of which appears in $W_{\mathcal{PAB}}$. The next lemma ties all three to two anticommutator defects, so that their simultaneous vanishing forces a uniform generation distribution. Let
\begin{equation}\label{rand:eq:defects}
\begin{aligned}
\Delta_B&=\|\{B_0,B_1\}\phi_1\|, & \Delta_C&=\|\{C_0,C_1\}\phi_1\|,\\
e_1&=\|(B_0C_1-\id)\phi_1\|, & e_1'&=\|(B_1C_0+\id)\phi_1\|.
\end{aligned}
\end{equation}
 
\begin{lemma}[moment reduction]\label{rand:lem:moments}
For any pure generation state $|\phi_1\rangle$,
\begin{equation}\label{rand:eq:momentbounds}
\begin{aligned}
\big|\avg{B_0}{1}\big|&\le \tfrac{\Delta_B}{2}+e_1', &
\big|\avg{C_0}{1}\big|&\le \tfrac{\Delta_C}{2}+e_1,\\
\big|\avg{B_0C_0}{1}\big|&\le \tfrac{\Delta_B}{2}+e_1', &&
\end{aligned}
\end{equation}
and consequently
\begin{equation}\label{rand:eq:Gbound}
\max_{b,c}\,p(b,c|s^*)\ \le\ \tfrac14+\tfrac18\,(2\Delta_B+\Delta_C)+\tfrac14\,(e_1+2e_1').
\end{equation}
In particular, if \eqref{rand:eq:S2} holds and $\Delta_B=\Delta_C=0$, then $p(b,c|s^*)=\tfrac14$ for every $(b,c)$.
\end{lemma}
 
\begin{proof}
$R=B_1C_0$ and $S=B_0C_1$ are Hermitian involutions, and the commutation relations \eqref{rand:eq:algebra} give the exact operator identities
\begin{equation}\label{rand:eq:conjugation}
\begin{aligned}
R\,B_0\,R&=\{B_0,B_1\}B_1-B_0,\\
S\,C_0\,S&=\{C_0,C_1\}C_1-C_0,\\
R\,(B_0C_0)\,R&=\big(\{B_0,B_1\}B_1-B_0\big)C_0 .
\end{aligned}
\end{equation}
For a Hermitian involution $R$ with $\|(R-\sigma\id)\psi\|\le e$, $\sigma=\pm1$, and any bounded $X$, writing $r=R\psi$ and $q=\sigma\psi$ so that $\|r-q\|\le e$,
\begin{equation}\label{rand:eq:stability}
\begin{split}
\big|\langle\psi|RXR|\psi\rangle-\langle\psi|X|\psi\rangle\big|
&=\big|\langle r-q|X|r\rangle+\langle q|X|r-q\rangle\big|\\
&\le 2e\,\|X\|.
\end{split}
\end{equation}
Taking $R=B_1C_0$ with $\sigma=-1$ and $X=B_0$, the identity \eqref{rand:eq:conjugation} gives $|2\avg{B_0}{1}-\langle\phi_1|\{B_0,B_1\}B_1|\phi_1\rangle|\le 2e_1'$, and Cauchy--Schwarz with $\|B_1\phi_1\|=1$ bounds the anticommutator term by $\Delta_B$, so $|\avg{B_0}{1}|\le\Delta_B/2+e_1'$. The same $R$ with $X=B_0C_0$, together with $\|B_1C_0\phi_1\|=1$, gives the third bound; and $S=B_0C_1$ with $\sigma=+1$, $X=C_0$ gives the second. Inserting \eqref{rand:eq:momentbounds} into \eqref{rand:eq:probform} and maximizing each sign yields \eqref{rand:eq:Gbound}. The bound \eqref{rand:eq:Gbound} is an upper estimate and is not claimed attainable, since the three signs $(-1)^b$, $(-1)^c$, $(-1)^{b+c}$ cannot be chosen independently. When $\Delta_B=\Delta_C=e_1=e_1'=0$ all three moments vanish and \eqref{rand:eq:probform} is exactly uniform.
\end{proof}
 
\subsection{Exact rigidity at the endpoint}
 
\begin{proposition}[endpoint rigidity]\label{rand:prop:rigidity}
In the representation \eqref{rand:eq:dilrep} with pure preparations, $W_{\mathcal{PAB}}=\WQ$ implies $P_{\mathrm{bad}}|\phi_0\rangle=P_{\mathrm{bad}}|\phi_1\rangle=0$, where $P_{\mathrm{bad}}$ projects onto the joint Jordan blocks on which the two local anticommutators do not both vanish. In particular $\Delta_B=\Delta_C=0$ and $p(b,c|s^*)=\tfrac14$.
\end{proposition}
 
\begin{proof}
Jordan's lemma \cite{Jordan1875,Masanes2006} applied to the pair of reflections $(B_0,B_1)$ on $\mathcal{H}_B$ yields an orthogonal decomposition into invariant subspaces of dimension at most two. On each two-dimensional block a basis can be chosen in which $B_0=Z$ and $B_1=\cos\theta_j\,Z+\sin\theta_j\,X$ with $\sin\theta_j\neq0$; blocks on which $B_0$ and $B_1$ commute split further into one-dimensional blocks, on which each observable is a sign $\pm1$. The same holds for $(C_0,C_1)$ on $\mathcal{H}_C$ with angles $\vartheta_k$. The joint blocks are $\mathcal{H}^B_j\otimes\mathcal{H}^C_k\otimes\mathcal{H}_R$, and their projectors $P_{jk}$ commute with every $B_y$ and $C_z$, hence with $\B$. On a joint block with both factors two-dimensional, $[A_0,A_1]=-2i\sin\theta_j\,Y$ and $[C_0,C_1]=2i\sin\vartheta_k\,Y$, so \eqref{rand:eq:Bsquared} gives $\B_{jk}^2=4\,\id-4\sin\theta_j\sin\vartheta_k\,(Y\otimes Y)\otimes\id_R$ and
\begin{equation}\label{rand:eq:blocknorm}
\|\B_{jk}\|=2\sqrt{1+|\sin\theta_j\sin\vartheta_k|},
\end{equation}
which equals $2\sqrt2$ if and only if $|\sin\theta_j|=|\sin\vartheta_k|=1$. If either local block is one-dimensional the corresponding commutator vanishes on it and $\|\B_{jk}\|\le2$. We call a joint block \emph{good} when $\|\B_{jk}\|=2\sqrt2$; on good blocks $\{B_0,B_1\}=2\cos\theta_j\,\id=0$ and $\{C_0,C_1\}=0$, whereas on every other block, including those with a one-dimensional factor, at least one of the two anticommutators is nonzero. Let $P_{\mathrm{good}}$ be the sum of the good joint block projectors and $P_{\mathrm{bad}}=\id-P_{\mathrm{good}}$.
 
First, $|\phi_2\rangle$ is supported on good blocks. Decomposing $|\phi_2\rangle=\sum_{jk}|\phi_2^{jk}\rangle$ along the joint blocks,
\begin{equation}\label{rand:eq:convexity}
\begin{split}
2\sqrt2=\avg{\B}{2}&=\sum_{jk}\langle\phi_2^{jk}|\B_{jk}|\phi_2^{jk}\rangle\\
&\le\sum_{jk}\|\B_{jk}\|\,\|\phi_2^{jk}\|^2\le2\sqrt2,
\end{split}
\end{equation}
and equality forces $\|\phi_2^{jk}\|=0$ on every block with $\|\B_{jk}\|<2\sqrt2$, that is, $P_{\mathrm{bad}}|\phi_2\rangle=0$.
 
Second, $|\phi_0\rangle$ and $|\phi_1\rangle$ are linearly independent. Otherwise $|\phi_1\rangle$ would satisfy the four relations \eqref{rand:eq:S1}--\eqref{rand:eq:S2} simultaneously; multiplying the two relations that involve $B_0$ by $B_0$ gives $(C_0-C_1)|\phi_1\rangle=0$, multiplying the two that involve $B_1$ by $B_1$ gives $(C_0+C_1)|\phi_1\rangle=0$, and adding yields $C_0|\phi_1\rangle=0$, contradicting the unitarity of $C_0$. Hence $\{|\phi_0\rangle,|\phi_1\rangle\}$ is a basis of $\mathcal{S}$ and $|\phi_2\rangle=\alpha|\phi_0\rangle+\beta|\phi_1\rangle$.
 
Third, $\alpha\neq0$ and $\beta\neq0$. From \eqref{rand:eq:S2} and $|\avg{B_yC_z}{1}|\le1$, the expectation of $\B$ on $|\phi_1\rangle$ is $-\avg{B_0C_0}{1}-1-1-\avg{B_1C_1}{1}\le0$, and from \eqref{rand:eq:S1} its expectation on $|\phi_0\rangle$ is $-1-\avg{B_0C_1}{0}+\avg{B_1C_0}{0}-1\le0$; if $|\phi_2\rangle$ were proportional to either vector, then $\avg{\B}{2}\le0$, contradicting $\avg{\B}{2}=2\sqrt2$.
 
Finally, the rank-two transfer. Set $v_0=P_{\mathrm{bad}}|\phi_0\rangle$ and $v_1=P_{\mathrm{bad}}|\phi_1\rangle$; applying $P_{\mathrm{bad}}$ to $|\phi_2\rangle=\alpha|\phi_0\rangle+\beta|\phi_1\rangle$ gives $\alpha v_0+\beta v_1=0$, so $v_1=-(\alpha/\beta)\,v_0$ with a nonzero factor. Since $P_{\mathrm{bad}}$ commutes with every $B_yC_z$, projecting \eqref{rand:eq:S1}--\eqref{rand:eq:S2} yields $B_0C_0v_0=v_0$, $B_1C_1v_0=v_0$, $B_0C_1v_0=v_0$, and $B_1C_0v_0=-v_0$. These are the same four incompatible relations as above, so $C_0v_0=0$ and $v_0=v_1=0$. Then $P_{\mathrm{good}}|\phi_1\rangle=|\phi_1\rangle$, and since $\{B_0,B_1\}$ and $\{C_0,C_1\}$ vanish on the good subspace, $\Delta_B=\Delta_C=0$. Lemma~\ref{rand:lem:moments} concludes.
\end{proof}
 
\begin{remark}
In infinite dimension the block sum becomes a direct integral and the equality argument in \eqref{rand:eq:convexity} applies fiberwise, but the tensor form of the commuting algebras is then an assumption rather than a theorem. All results in this appendix are stated for finite-dimensional tensor receivers, as in Definition~\ref{rand:def:eve}; the block operator bound of Lemma~\ref{rand:lem:bound}, in contrast, uses only the algebra \eqref{rand:eq:algebra} and holds in any dimension.
\end{remark}
 
\begin{lemma}[mixed preparations at the endpoint]\label{rand:lem:mixed}
In the representation \eqref{rand:eq:dilrep} with preparations $\rho_x$ possibly mixed, $W_{\mathcal{PAB}}=\WQ$ implies $p(b,c|s^*)=\tfrac14$.
\end{lemma}
 
\begin{proof}
Let $K_x=V^\dagger M_x V$, so that $W^{[x]}=\tr(\rho_x K_x)$. By the block bounds of Lemma~\ref{rand:lem:bound}, $M_x\preceq\lambda_x\id$, and conjugation by an isometry is positive and unital on the qubit, so $K_x\preceq\lambda_x\id_2$ and $A_x=\lambda_x\id_2-K_x\succeq0$. Saturation $W^{[x]}=\lambda_x$ gives $\tr(\rho_xA_x)=0$ with both factors positive semidefinite, hence $A_x\rho_x=0$ and $\operatorname{supp}\rho_x\subseteq\ker A_x$. Writing $\rho_x=\sum_k p_{x,k}|\psi_{x,k}\rangle\langle\psi_{x,k}|$ with $p_{x,k}>0$, every component satisfies $K_x|\psi_{x,k}\rangle=\lambda_x|\psi_{x,k}\rangle$, so $|\phi_{x,k}\rangle=V|\psi_{x,k}\rangle\in\mathcal{S}$ saturates block $x$ and obeys the corresponding saturation relations of Lemma~\ref{rand:lem:bound}; in particular $\langle\phi_{2,\ell}|\B|\phi_{2,\ell}\rangle=2\sqrt2$ for every component of $\rho_2$. Now fix an arbitrary component $|\phi_{1,k}\rangle$ of the generation preparation and arbitrary components $|\phi_{0,j}\rangle$, $|\phi_{2,\ell}\rangle$ of the other two. The three vectors lie in $\mathcal{S}$ and satisfy their block relations, so the four paragraphs of the proof of Proposition~\ref{rand:prop:rigidity} apply verbatim to this triple and give $\{B_0,B_1\}|\phi_{1,k}\rangle=\{C_0,C_1\}|\phi_{1,k}\rangle=0$, hence, by Lemma~\ref{rand:lem:moments}, uniform outcome statistics for each $k$. Averaging over $k$ gives $p(b,c|s^*)=\tfrac14$ for the mixed generation state.
\end{proof}
 
\subsection{Proof of Theorem~\ref{rand:thm:main}}
 
For a classical adversary, the independence $q(\lambda|x,y,z)=q(\lambda)$ makes $W_{\mathcal{PAB}}$ affine over the adversarial decomposition, and each branch obeys $W_{\mathcal{PAB}}[p_\lambda]\le\WQ$ by the block bounds of Lemma~\ref{rand:lem:bound}. Hence $W_{\mathcal{PAB}}[p]=\WQ$ forces $W_{\mathcal{PAB}}[p_\lambda]=\WQ$ on every branch of nonzero weight. Dilating each branch by Lemma~\ref{rand:lem:dilation}, which is exact and requires no saturation hypothesis, Lemma~\ref{rand:lem:mixed} gives $p_\lambda(b,c|s^*)=\tfrac14$ for all $(b,c)$ on every branch. Therefore $G(BC|\Lambda,s^*;\WQ)=\sum_\lambda q_\lambda\cdot\tfrac14=\tfrac14$ and $\Hmin=2$. The explicit endpoint strategy of Lemma~\ref{rand:lem:bound} shows the feasible set at $w=\WQ$ is nonempty, so the statement is not vacuous. \hfill$\blacksquare$
 
\subsection{Symmetry and the four-input orbit}
 
The witness \eqref{eq:CC2} is the coefficient tensor $\gamma_{xyz}$ on the correlators $\corr{y}{z}{x}$, with $\gamma_{000}=\gamma_{011}=\gamma_{101}=2$, $\gamma_{110}=-2$, $\gamma_{200}=\gamma_{201}=\gamma_{211}=-1$, $\gamma_{210}=1$, and all other entries zero; no single-party terms appear. A relabeling of the experiment is a tuple $(\pi,\sigma_B,\sigma_C,u,v;\text{swap})$, with $\pi\in S_3$ acting on preparations, $\sigma_B,\sigma_C\in S_2$ on settings, outcome flips $u_y,v_z\in\{0,1\}$ for Bob and Charlie, and an optional exchange of the two receivers. Elements without the exchange act on correlators as
\begin{equation}\label{rand:eq:action}
\corr{y}{z}{x}\ \longmapsto\ (-1)^{u_y+v_z}\,\corr{\sigma_B(y)}{\sigma_C(z)}{\pi(x)},
\end{equation}
and transport the generation inputs as $(x,y,z)\mapsto(\pi(x),\sigma_B(y),\sigma_C(z))$. Elements with the exchange act as the receiver transposition $(y,z)\mapsto(z,y)$ \emph{followed} by \eqref{rand:eq:action}, that is,
\begin{equation}\label{rand:eq:actionswap}
\corr{y}{z}{x}\ \longmapsto\ (-1)^{u_z+v_y}\,\corr{\sigma_B(z)}{\sigma_C(y)}{\pi(x)},
\end{equation}
with input transport $(x,y,z)\mapsto(\pi(x),\sigma_B(z),\sigma_C(y))$. The composition order matters for the exchange elements, and Table~\ref{rand:tab:autos} refers to the convention \eqref{rand:eq:actionswap}; under the opposite order the flip vectors of elements 9, 12, 14, and 15 must be modified. We also note that the induced action on the twelve two-point correlators is not faithful: element 4 of Table~\ref{rand:tab:autos}, the simultaneous flip of all outputs of both receivers, acts trivially on every $\corr{y}{z}{x}$ but nontrivially on behaviors. All group statements below refer to the action on behaviors.
 
\begin{table}[t]
\caption{\label{rand:tab:autos}The sixteen relabelings that leave the coefficient tensor of $W_{\mathcal{PAB}}$ invariant, in the convention of Eqs.~\eqref{rand:eq:action}--\eqref{rand:eq:actionswap}, obtained by exhaustive enumeration of the $768$ candidates. Elements 2, 5, and 10 generate the group.}
\begin{ruledtabular}
\begin{tabular}{ccccccc}
\# & $\pi$ & $\sigma_B$ & $\sigma_C$ & $u=(u_0,u_1)$ & $v=(v_0,v_1)$ & swap\\
\hline
1  & id     & id     & id     & (0,0) & (0,0) & no\\
2  & id     & id     & id     & (0,1) & (0,1) & yes\\
3  & id     & id     & id     & (1,0) & (1,0) & yes\\
4  & id     & id     & id     & (1,1) & (1,1) & no\\
5  & id     & (01)   & (01)   & (0,0) & (0,0) & yes\\
6  & id     & (01)   & (01)   & (0,1) & (0,1) & no\\
7  & id     & (01)   & (01)   & (1,0) & (1,0) & no\\
8  & id     & (01)   & (01)   & (1,1) & (1,1) & yes\\
9  & (01)   & id     & (01)   & (0,0) & (0,1) & yes\\
10 & (01)   & id     & (01)   & (0,1) & (0,0) & no\\
11 & (01)   & id     & (01)   & (1,0) & (1,1) & no\\
12 & (01)   & id     & (01)   & (1,1) & (1,0) & yes\\
13 & (01)   & (01)   & id     & (0,0) & (1,0) & no\\
14 & (01)   & (01)   & id     & (0,1) & (1,1) & yes\\
15 & (01)   & (01)   & id     & (1,0) & (0,0) & yes\\
16 & (01)   & (01)   & id     & (1,1) & (0,1) & no\\
\end{tabular}
\end{ruledtabular}
\end{table}
 
\begin{proposition}[symmetry transport]\label{rand:prop:symmetry}
The relabelings that leave $\gamma$ invariant form a group of order 16, listed in Table~\ref{rand:tab:autos} and generated by the elements 2, 5, and 10, whose action on the twelve generation inputs has the three orbits \eqref{rand:eq:orbits}. If such a relabeling $T$ sends $s\mapsto s'$, then $G(BC|\Lambda,s;w)=G(BC|\Lambda,s';w)$ for every $w\le\WQ$, against the adversaries of Definition~\ref{rand:def:eve}. Consequently Theorem~\ref{rand:thm:main} holds for every $s\in O_B$.
\end{proposition}
 
\begin{proof}
The group and its orbits are a finite computation. Of the $|S_3|\cdot|S_2|^2\cdot2^4\cdot2=768$ candidate relabelings, exactly the sixteen of Table~\ref{rand:tab:autos} map $\gamma$ to itself, their action on the inputs produces the three orbits \eqref{rand:eq:orbits} in the convention stated above, and the subgroup generated by elements 2, 5, and 10 exhausts the sixteen elements. The guessing probability is invariant because a relabeling maps PAB behaviors to PAB behaviors of the same class: permutations of preparations, settings, and outcomes act on the classical processing of a PAB realization, and the receiver exchange maps a realization to another with $\mathcal{H}_B$ and $\mathcal{H}_C$ interchanged, in all cases preserving the message bound $d=2$, the tensor structure, and finite-dimensionality. Exact invariance of $\gamma$ gives $W_{\mathcal{PAB}}[Tp]=W_{\mathcal{PAB}}[p]$ for every behavior $p$, and the output relabeling permutes the four output pairs at the transported input, so $\max_{b,c}p(b,c|s)=\max_{b',c'}(Tp)(b',c'|s')$. Hence any feasible decomposition $\{q_\lambda,p_\lambda\}$ at $(s,w)$ maps to a feasible decomposition $\{q_\lambda,Tp_\lambda\}$ at $(s',w)$ with the same guessing value, and conversely with $T^{-1}$; taking suprema gives the equality of the two guessing probabilities. Explicit maps carrying $(1,0,0)$ across $O_B$ are elements 5, 9, and 13 of Table~\ref{rand:tab:autos}, with
\begin{equation}\label{rand:eq:orbitmaps}
(1,0,0)\ \mapsto\ (1,1,1),\ (0,0,1),\ (0,1,0),
\end{equation}
respectively. Combined with Theorem~\ref{rand:thm:main} at $s^*=(1,0,0)$, this proves the endpoint statement for all of $O_B$.
\end{proof}
 
\begin{remark}[the orbit $O_A$]\label{rand:rem:OA}
At $W_{\mathcal{PAB}}=\WQ$, every adversarial branch obeys the saturation relations of Lemma~\ref{rand:lem:bound}, applied componentwise as in Lemma~\ref{rand:lem:mixed}, so $\corr{0}{0}{0}=\corr{1}{1}{0}=\corr{0}{1}{1}=1$ and $\corr{1}{0}{1}=-1$ on every branch. At each $s\in O_A$ the output pair is therefore perfectly correlated or anticorrelated, $\max_{b,c}p_\lambda(b,c|s)\ge\tfrac12$, and $\Hmin(BC|\Lambda,s)\le1$. The two-bit certificate is thus a property of the orbit $O_B$, not of the maximal violation alone.
\end{remark}
 
\subsection{Decoupling of the channel environment at the endpoint}
 
\begin{lemma}[environment decoupling]\label{rand:lem:decoupling}
In the representation \eqref{rand:eq:dilrep} with pure preparations and $W_{\mathcal{PAB}}=\WQ$, the conditional state of the Stinespring environment $\mathcal{H}_R$ given the outputs at $s^*$ is outcome-independent: writing $\rho_R^{(b,c)}$ for the state of $\mathcal{H}_R$ conditioned on the pair $(b,c)$,
\begin{equation}\label{rand:eq:decouple}
p(b,c|s^*)\;\rho_R^{(b,c)}=\tfrac14\,\rho_R\qquad\text{for all }(b,c).
\end{equation}
In particular $\Hmin(BC|R,s^*)=2$ in this restricted pure-input model.
\end{lemma}
 
\begin{proof}
By Proposition~\ref{rand:prop:rigidity}, $|\phi_1\rangle$ is supported on the good joint blocks, and on each such block the Jordan bases give $B_0=Z\otimes\id$, $B_1=\sigma^B_{j}\,X\otimes\id$, $C_0=\id\otimes Z$, $C_1=\sigma^C_{k}\,\id\otimes X$ with signs $\sigma^B_j,\sigma^C_k\in\{\pm1\}$, acting on $\mathbb{C}^2\otimes\mathbb{C}^2\otimes\mathcal{H}_R$. Since the block projectors commute with every $B_yC_z$, the relations \eqref{rand:eq:S2} project blockwise: the component $|\psi_{jk}\rangle=P_{jk}|\phi_1\rangle$ satisfies $S_1|\psi_{jk}\rangle=|\psi_{jk}\rangle$ and $S_2|\psi_{jk}\rangle=|\psi_{jk}\rangle$ for the two commuting and independent two-qubit Pauli operators $S_1=Z\otimes\sigma^C_kX$ and $S_2=-\sigma^B_jX\otimes Z$, both acting trivially on $\mathcal{H}_R$. Their common $+1$ eigenspace in $\mathbb{C}^4$ is one-dimensional, spanned by a stabilizer state $|\chi_{jk}\rangle$, so
\begin{equation}\label{rand:eq:factorize}
|\psi_{jk}\rangle=|\chi_{jk}\rangle\otimes|\eta_{jk}\rangle,\qquad
\sum_{jk}\||\eta_{jk}\rangle\|^2=1 .
\end{equation}
The generation observables restrict on each good block to $Z\otimes Z$, $Z\otimes\id$, and $\id\otimes Z$, each of which anticommutes with $S_1$ or $S_2$, so all three have zero expectation on $|\chi_{jk}\rangle$ and $\langle\chi_{jk}|\Pi^B_{b|0}\Pi^C_{c|0}|\chi_{jk}\rangle=\tfrac14$ for every $(b,c)$. Since the measurement projectors preserve the blocks, cross-block terms vanish in the partial trace over $\mathcal{H}_B\otimes\mathcal{H}_C$, and
\begin{equation}\label{rand:eq:decoupleproof}
\begin{split}
p(b,c|s^*)\,\rho_R^{(b,c)}
&=\sum_{jk}\langle\chi_{jk}|\Pi^B_{b|0}\Pi^C_{c|0}|\chi_{jk}\rangle\,
|\eta_{jk}\rangle\langle\eta_{jk}|\\
&=\tfrac14\sum_{jk}|\eta_{jk}\rangle\langle\eta_{jk}|,
\end{split}
\end{equation}
independent of $(b,c)$.
\end{proof}

\begin{corollary}[Bell locality of the generation slice]\label{rand:cor:local}
In the representation \eqref{rand:eq:dilrep} with pure preparations and $W_{\mathcal{PAB}}=\WQ$, the conditional behavior at $x=1$ has vanishing single-party marginals,
\begin{equation}\label{rand:eq:localmarginals}
\avg{B_0}{1}=\avg{B_1}{1}=\avg{C_0}{1}=\avg{C_1}{1}=0,
\end{equation}
and correlators
\begin{equation}\label{rand:eq:localcorrelators}
\corr{0}{1}{1}=1,\qquad \corr{1}{0}{1}=-1,\qquad\\
\corr{0}{0}{1}=\corr{1}{1}{1}=0 .
\end{equation}
Every CHSH expression evaluated on this behavior has modulus at most the local bound $2$, and the behavior is reproduced exactly by an explicit local hidden-variable model. The same conclusion holds for possibly mixed preparations, by the componentwise argument of Lemma~\ref{rand:lem:mixed}.
\end{corollary}

\begin{proof}
The values $\corr{0}{1}{1}=1$ and $\corr{1}{0}{1}=-1$ are the saturation relations \eqref{rand:eq:S2}, while $\avg{B_0}{1}=\avg{C_0}{1}=\corr{0}{0}{1}=0$ follow from Proposition~\ref{rand:prop:rigidity} together with Lemma~\ref{rand:lem:moments}. For the three remaining moments we use the block decomposition of the proof of Lemma~\ref{rand:lem:decoupling}, in which $|\phi_1\rangle$ is supported on the good joint blocks and the Jordan bases give $B_0=Z\otimes\id$, $B_1=\sigma^B_j\,X\otimes\id$, $C_0=\id\otimes Z$, and $C_1=\sigma^C_k\,\id\otimes X$, with the component $|\psi_{jk}\rangle=P_{jk}|\phi_1\rangle$ factorizing as $|\chi_{jk}\rangle\otimes|\eta_{jk}\rangle$ for the stabilizer state $|\chi_{jk}\rangle$ fixed by $S_1=Z\otimes\sigma^C_kX$ and $S_2=-\sigma^B_jX\otimes Z$. On such a block the operators $B_1C_1=\sigma^B_j\sigma^C_k\,X\otimes X$, $B_1=\sigma^B_j\,X\otimes\id$, and $C_1=\sigma^C_k\,\id\otimes X$ each anticommute with $S_1$ or with $S_2$, hence have vanishing expectation on $|\chi_{jk}\rangle$, and since the block projectors commute with $B_1$ and with $C_1$ no cross-block term survives, which gives $\corr{1}{1}{1}=\avg{B_1}{1}=\avg{C_1}{1}=0$ and completes \eqref{rand:eq:localmarginals}--\eqref{rand:eq:localcorrelators}. Since all marginals vanish, the four inequivalent CHSH functionals evaluate on \eqref{rand:eq:localcorrelators} to $0$, $2$, $-2$ and $0$, so the maximum modulus over the eight relabelings equals $2$ and the behavior sits on the boundary of the local set. For the explicit model, let $\beta_y$ and $\gamma_z$ in $\{\pm1\}$ denote the deterministic values assigned to the reflections $B_y$ and $C_z$, related to the outputs by $\beta_y=(-1)^{b}$ and $\gamma_z=(-1)^{c}$, and take the uniform mixture of the four assignments in which $\beta_0$ and $\beta_1$ are drawn independently while $\gamma_1=\beta_0$ and $\gamma_0=-\beta_1$. Then $\corr{0}{1}{1}=\langle\beta_0\gamma_1\rangle=\langle\beta_0^2\rangle=1$ and $\corr{1}{0}{1}=\langle\beta_1\gamma_0\rangle=-\langle\beta_1^2\rangle=-1$, whereas $\corr{0}{0}{1}=-\langle\beta_0\beta_1\rangle$ and $\corr{1}{1}{1}=+\langle\beta_0\beta_1\rangle$, and the product $\beta_0\beta_1$ is itself uniformly distributed on $\{\pm1\}$ because the two signs are drawn independently, so both correlators vanish. The four marginals vanish for the same reason, and the model reproduces \eqref{rand:eq:localmarginals}--\eqref{rand:eq:localcorrelators} exactly. \end{proof}
 
\begin{remark}[scope]\label{rand:rem:scope}
Lemma~\ref{rand:lem:decoupling} is a statement about the Stinespring dilation of the broadcast channel in the pure-input model. It is not a security statement against a quantum adversary who holds a purification $|\Psi_x\rangle_{QE}$ of the preparations with only $Q$ sent to the channel: in that case the relevant global states $(V\otimes\id_E)|\Psi_x\rangle$ need not lie in a two-dimensional subspace, the rank-two transfer of Proposition~\ref{rand:prop:rigidity} does not apply, and quantum-adversary security remains open.
\end{remark}

\end{document}